\def\comments{1}

\documentclass[11pt]{article}
\usepackage[utf8]{inputenc}

\usepackage{graphicx}
\usepackage{amsmath,amssymb,amsthm}
\usepackage[margin=1.2in]{geometry}
\usepackage{microtype}
\usepackage{indentfirst}
\usepackage{verbatim}

\usepackage[toc,page]{appendix}

\ifnum\comments=1
\newcommand{\mynote}[2]{{\marginpar{\color{#1} \tiny #2}}}
\else
\newcommand{\mynote}[2]{}
\fi

\usepackage{authblk}

\usepackage{bbm} 
\usepackage{soul} 

\usepackage[colorlinks,citecolor=blue,urlcolor=blue,linkcolor=blue]{hyperref}

\let\originalleft\left
\let\originalright\right
\renewcommand{\left}{\mathopen{}\mathclose\bgroup\originalleft}
\renewcommand{\right}{\aftergroup\egroup\originalright}

\newcommand{\uni}{\mathcal{X}} 
\newcommand{\usize}{T} 
\newcommand{\elem}{x} 
\newcommand{\queries}{Q} 
\newcommand{\qmat}{W} 
\newcommand{\qent}{w} 
\newcommand{\concepts}{\mathcal{C}} 
\newcommand{\concept}{c}
\newcommand{\cmat}{D} 
\newcommand{\mech}{\mathcal{M}} 
\newcommand{\post}{\mathcal{A}} 
\newcommand{\err}{\mathrm{err}} 
\newcommand{\errz}{\mathrm{err}^{\ell_2}}
\newcommand{\errinf}{\mathrm{err}^{\ell_\infty}} 
\renewcommand{\sc}{\mathrm{sc}}
\newcommand{\scz}{\mathrm{sc}^{\ell_2}}  
\newcommand{\sczinf}{\mathrm{sc}^{\ell_\infty}}  
\newcommand{\fnorm}{\ensuremath{\gamma_F}} 
\newcommand{\fnorminf}{\ensuremath{\gamma_2}} 
\newcommand{\query}{q}
\newcommand{\qsize}{k}
\newcommand{\ds}{X}
\newcommand{\dsrow}{x} 
\newcommand{\dist}{\mu}
\newcommand{\hist}{h}
\newcommand{\dsize}{n}
\newcommand{\priv}{\eps}
\newcommand{\privd}{\delta}

\newcommand{\trans}{\mathcal{T}_\mech} 
\newcommand{\outspace}{\Omega} 
\newcommand{\loutspace}{\bar{\Omega}} 
\newcommand{\loss}{\Lambda}

\newcommand{\inner}{d} 

\newcommand{\gauss}{\sigma_{\eps, \delta}} 

\newcommand{\dista}{\lambda} 
\newcommand{\distb}{\mu} 

\newcommand{\from}{:}
\newcommand{\set}[1]{\left\{#1\right\}}

\newcommand{\R}{\mathbb{R}}
\newcommand{\N}{\mathbb{N}}
\newcommand{\Z}{\mathbb{Z}}
\newcommand{\E}{\mathbb{E}}

\newcommand{\I}{\mathbb{I}}
\newcommand{\tr}{\mathrm{Tr}}
\newcommand{\eps}{\varepsilon}

\newcommand{\zo}{\{0,1\}}

\renewcommand{\div}{\mathrm{D}_{\text{KL}}}
\newcommand{\chd}{\mathrm{D}_{\chi^2}}

\newcommand{\mypar}[1]{\medskip\noindent\textbf{#1.}}

\newcommand{\ex}[2]{{\ifx&#1& \mathbb{E} \else \underset{#1}{\mathbb{E}} \fi \left[#2\right]}}
\newcommand{\pr}[2]{{\ifx&#1& \mathbb{P} \else \underset{#1}{\mathbb{P}} \fi \left[#2\right]}}

\newtheorem{theorem}{Theorem}

\newtheorem{lemma}[theorem]{Lemma}

\newtheorem{corollary}[theorem]{Corollary}
\newtheorem{remark}[theorem]{Remark}

\newtheorem{definition}[theorem]{Definition}

\title{The Power of Factorization Mechanisms \\ in Local and Central Differential Privacy}
\author[1]{Alexander Edmonds}
\author[1]{Aleksandar Nikolov}
\author[2]{Jonathan Ullman}
\affil[1]{Department of Computer Science, University of Toronto}
\affil[2]{Khoury College of Computer Sciences, Northeastern University}
\date{\today}

\begin{document}

\maketitle

\begin{abstract}
	We give new characterizations of the sample complexity of answering linear queries (statistical queries) in the local and central models of differential privacy: 
	\begin{itemize}
		\item In the non-interactive local model, we give the first approximate characterization of the sample complexity. Informally our bounds are tight to within polylogarithmic factors in the number of queries and desired accuracy.  Our characterization extends to agnostic learning in the local model.  
		\item In the central model, we give a characterization of the sample complexity in the high-accuracy regime that is analogous to that of Nikolov, Talwar, and Zhang (STOC 2013), but is both quantitatively tighter and has a dramatically simpler proof.  
\end{itemize}
 	Our lower bounds apply equally to the empirical and population estimation problems.
	In both cases, our characterizations show that a particular
        factorization mechanism is approximately optimal, and the optimal sample complexity is bounded from above and below by well studied factorization norms of a matrix associated with the queries. 
\end{abstract}

\section{Introduction}

Differential privacy \cite{DworkMNS06} is a rigorous mathematical
framework for protecting individual privacy that is well suited to
statistical data analysis.  In addition to a rich academic literature,
differential privacy is now being deployed on a large scale by Apple
\cite{AppleDP17}, Google \cite{ErlingssonPK14, Bittau+17, Wilson+19},
Uber~\cite{JohnsonNS18}, and the US Census Bureau \cite{CensusDP17}.

To compute statistics of the data with differential privacy---or any notion of privacy---we have to inject noise into the computation of these statistics~\cite{DinurN03}.  The amount of noise is highly dependent on the particular statistic, and thus a central problem in differential privacy is to determine how much error is necessary to compute a given statistic.

In this work we consider the class of \emph{linear queries} (also called \emph{statistical queries}~\cite{Kearns93}).  The simplest example of a linear query is ``What fraction of individuals in the data have property $P$?''  Workloads of linear queries capture a variety of statistical tasks: computing histograms and PDFs, answering range queries and computing CDFs, estimating the mean, computing correlations and higher-order marginals, and estimating the risk of a classifier.  

The power of differentially private algorithms for answering a \emph{worst-case} workload of linear queries is well understood~\cite{BunUV14}, and known bounds are essentially tight as a function of the dataset size, the data domain, and the size of the workload.  However, many workloads, such as those corresponding to computing PDFs or CDFs, have additional structure that makes it possible to answer them with less error than these worst-case workloads.  Thus, a central question is
\begin{quotation}
	\noindent\emph{Can we characterize the amount of error required to estimate a given workload of linear queries subject to differential privacy in terms of natural properties of the workload, and can we achieve this error via computationally efficient algorithms?}
\end{quotation}

In the central model, there has been dramatic progress on this question~\cite{HardtT10, BhaskaraDKT12, NikolovTZ13, Nikolov15, BlasiokBNS19}, giving approximate characterizations for every workload of linear queries.  We extend this line of work in two ways:
\begin{enumerate}
	\item We give the first approximate characterization for the \emph{non-interactive local model of differential privacy}~\cite{DworkMNS06, KasiviswanathanLNRS08}.  This result is also much sharper than analogous results for the central model of differential privacy.  
	
	\item We give a new approximate characterization for the
          \emph{central model of differential privacy} in the
          high-accuracy regime (equivalently, in the large-dataset regime).
 	This characterization is analogous to a result of~\cite{NikolovTZ13}, but it is quantitatively tighter and its proof is dramatically simpler.  For $\ell_2^2$ error, our characterization is tight up to a constant factor.
\end{enumerate}

In particular, our results show that a natural and well studied type
of \emph{factorization mechanism} is approximately optimal in these
settings.  Factorization mechanisms capture a number of special-purpose
mechanisms from the theory
literature~\cite{BarakCDKMT07,DworkNPR10,ChanSS11,ThalerUV12,ChandrasekaranTUW14},
were involved in previous characterizations, and also roughly capture
the \emph{matrix mechanisms}~\cite{LiHRMM10,McKennaMHM18} from the
databases literature, which have been developed into practical
algorithms for US Census Data.\footnote{In a nutshell,
  the matrix mechanism is a particular factorization mechanism
  designed for the special case of $\ell_2^2$ error, and combined with various
  optimizations and post-processing techniques to improve
  computational efficiency and utility. Usually the matrix mechanism
  is presented in the special case of pure differential privacy.}

Our characterization in the local model extends to agnostic PAC
learning, and shows that the optimal learner for any family of queries
is to use the optimal factorization mechanism to estimate the error of
every concept.  Our characterization is sharper than the previous
characterization of~\cite{KasiviswanathanLNRS08}, which loses
polynomial factors in the SQ dimension~\cite{BlumFJKMR94}.

\subsection{Background: Linear Queries and Factorization Mechanisms}
We start by briefly introducing the relevant concepts and definitions
necessary to state our results. See Section~\ref{sec:prelims} for a
more thorough treatment of the necessary background. 

\mypar{Linear Queries} Suppose we are given a \emph{dataset} $\ds = (\dsrow_1,\dots,\dsrow_n) \in \uni^n$, where each entry $\dsrow_i$ is the data of one individual and $\uni$ is some \emph{data universe}.  We will treat the size of the dataset $n$ as public information.  A linear query is specified by a bounded function $\query \from \uni \to \R$ and (abusing notation) its answer is $\query(\ds) = \frac{1}{n} \sum_{i=1}^{n} \query(\dsrow_i)$.  A \emph{workload} is a set of linear queries $\queries = \set{\query_1,\dots,\query_\qsize}$, and we use $\queries(\ds) = (\query_1(\ds),\dots,\query_\qsize(\ds))$ to denote the answers.

Given a workload of queries, we can associate a \emph{workload matrix} $\qmat \in \R^{\queries \times \uni}$, defined by $\qmat_{\query,\elem}= \query(\elem)$.  The convention of calling the above queries ``linear'' stems from the fact that they can be written as the product of the workload matrix with the histogram vector of the dataset.  As such, we will sometimes use $\queries$ and $\qmat$ interchangeably.

\mypar{Error and Sample Complexity} 
Our goal is to design an $(\eps,\delta)$-differentially private mechanism $\mech$ that takes a dataset $\ds$ and accurately estimates $\queries(\ds)$ for an appropriate measure of accuracy.  In this work we primarily consider accuracy in the $\ell_\infty$ norm, and define
\begin{align*}
&\err^{\ell_\infty}(\mech,\queries, n) = \max_{\ds \in \uni^n} \ex{\mech}{\| \mech(\ds) - \queries(\ds) \|_\infty}, \\
&\err^{\ell_\infty}_{\eps, \delta}(\queries, n) = \min_{\textrm{$(\eps,\delta)$-DP $\mech$}} \err^{\ell_\infty}(\mech,\queries,n).
\end{align*}
Privacy becomes easier to achieve as the dataset size $n$ grows.  We are interested in the \emph{sample complexity}, which is the smallest size of dataset on which it is possible to achieve a specified error $\alpha$ for given privacy parameters $\priv$ and $\privd$:
\begin{align*}
\sc^{\ell_\infty}_{\eps,\delta}(\queries,\alpha) = \min\left\{ n : \err^{\ell_\infty}_{\eps,\delta}(\queries,n) \leq \alpha \right\}.
\end{align*}

\mypar{The Approximate Factorization Mechanisms}  One of the most
basic tools in the central-model of differential privacy is the \emph{Gaussian mechanism} (see e.g.~\cite{DworkR14}).  This mechanism computes the vector of answers to the queries $\queries(\ds)$ and perturbs it with spherical Gaussian noise scaled to the \emph{$\ell_2$-sensitivity} of the workload.  In particular, the sample complexity of this mechanism is
$$
O\left( \frac{\| \qmat \|_{1 \to 2} \sqrt{\log(1/\delta) \log \qsize}}{ \priv \alpha}\right).
$$
where $\| \qmat \|_{1 \to 2}$ denotes the largest $\ell_2$ norm of any column of $\qmat$, which is the $\ell_2$-sensitivity.

One can try to improve this mechanism by replacing $\qmat$ with a simpler workload of queries $A$, and then attempting to reconstruct the answer to $\qmat$ by applying a linear transform $R$ such that $\qmat = RA$.  One can show that the overall mechanism has error $\| R \|_{2 \to \infty} \| A \|_{1 \to 2}$, where $\|R \|_{2 \to \infty}$ denotes the maximum $\ell_2$ norm of any row of $R$.  This quantity can be dramatically smaller than $\| \qmat \|_{1 \to \infty}$, for example if $\qmat$ contains many copies of the same query.

The \emph{factorization mechanism} chooses the optimal factorization $\qmat = RA$, giving error proportional to the \emph{factorization norm}
$$
\fnorminf(\qmat) = \min \{ \| R \|_{2 \to \infty} \| A \|_{1 \to 2} : \qmat = RA \}.
$$
The sample complexity of this mechanism is thus 
$$
\sczinf(\mech_{\fnorminf},\queries,\alpha) = O\left(\frac{\fnorminf(\qmat) \sqrt{\log(1/\delta) \log|\queries|}}{\priv \alpha}\right).
$$
We note that that the factorization norm $\fnorminf(\qmat)$ and an optimal factorization $\qmat = RA$ can be computed in time polynomial in the size of $\qmat$ via semidefinite programming~\cite{LinialS09}.

Finally, we can try to further improve the mechanism using an \emph{approximate factorization mechanism} that approximates the workload $\qmat$ with a simpler workload $\widetilde\qmat$ that is entrywise close to $\qmat$, and applying the factorization mechanism to $\widetilde\qmat$.  The error of this mechanism is proportional to the \emph{approximate factorization norm}
$$
\fnorminf(\qmat,\alpha) = \min \{ \fnorminf(\widetilde\qmat) : \| \qmat - \widetilde\qmat \|_{1 \to \infty} \leq \alpha/2 \},
$$
where $\| \qmat - \widetilde\qmat \|_{1 \to \infty}$ is the maximum
absolute difference between entries of $\qmat$ and $\widetilde{\qmat}$.
The sample complexity of this mechanism is thus
$$
\sczinf(\mech_{\fnorminf, \alpha},\queries,\alpha) = O\left(\frac{\fnorminf(\qmat, \alpha/2) \sqrt{\log(1/\delta) \log|\queries|}}{\priv \alpha}\right).
$$

\mypar{The Local Model}  Although we have discussed the factorization mechanism in the context of central differential privacy, these ideas can all be adapted to \emph{(non-interactive) local differential privacy}.  In this model, each user will apply a separate $(\priv,\privd)$-differentially private mechanism $\mech_1,\dots,\mech_n$ to their own data, and the output can then be postprocessed using an arbitrary algorithm $\mathcal{A}$, so the mechanism can be expressed as
$$
\mech(\ds) = \post(\mech_1(\ds_1),\dots,\mech_n(\ds_n))
$$
We define $\err^{\ell_\infty,\textrm{loc}}_{\eps,\delta}$, and $\sc^{\ell_\infty,\textrm{loc}}_{\eps,\delta}$ analogously to the central model, but with the minimum taken over mechanisms that are $(\eps, \delta)$-DP in the local model.

Since the queries are linear, we can simply have each user apply the approximate factorization mechanism to their own data and average the results.  One can show that randomizing each individual's data independently increases the variance of the noise by a factor of $\sqrt{n}$ compared to the central model version of the mechanism.  One can also achieve $(\priv,0)$-differential privacy by replacing Gaussian noise with a different subgaussian noise distribution.  Putting it together, the resulting sample complexity becomes
\begin{equation} \label{eq:lafm-ub}
\sczinf(\mech_{\fnorminf, \alpha}^{\textrm{loc}},\queries,\alpha) = O\left(\frac{\fnorminf(\qmat, \alpha/2)^2 \log |\queries|}{ \priv^2 \alpha^2}\right).
\end{equation}

\subsection{Our Results}

\subsubsection{Linear Queries in the Local Model}

Our main result in the local model shows that the approximate factorization mechanism described above is approximately optimal among all non-interactive locally differentially private mechanisms.
\begin{theorem} [Informal] \label{thm:main-local}
Let $\alpha, \priv, \privd > 0$ be smaller than some absolute constants and let $\queries$ be a workload of linear queries with workload matrix $\qmat$.  Then, for some $\alpha' = \Omega( \alpha/ \log(1/\alpha) )$,
$$
\sc_{\priv,0}^{\ell_\infty, \mathrm{loc}}(\queries, \alpha') = \Omega\left( \frac{\fnorminf(\qmat, \alpha/2)^2 }{ \priv^2 \alpha^2 }\right).
$$
\end{theorem} 
To interpret the theorem, it helps to start by imagining that $\fnorminf(\qmat,\alpha'/2) = \fnorminf(\qmat,\alpha/2)$, in which case the theorem would show that the sample complexity of answering queries up to error $\alpha'$ is
$$
\Omega\left( \frac{\fnorminf(\qmat, \alpha'/2)^2 }{ \priv^2 \alpha^2} \right),
$$
which differs from the sample complexity of the local approximate
factorization mechanism, given in~\eqref{eq:lafm-ub}, by a factor of
just $O(\log(1/\alpha')^2\log |\queries|)$.  The fact that we take
$\alpha' < \alpha$ means that $\fnorminf(\qmat, \alpha/2)$ can be much
smaller than $\fnorminf(\qmat,\alpha'/2)$.\footnote{For example, if
  every entry of $\qmat$ is at most $\alpha$ in absolute value, then
  $\fnorminf(\qmat, \alpha) = 0$ whereas $\fnorminf(\qmat, \alpha')$
  can be arbitrarily large for $\alpha' < \alpha$, but this behavior
  typically does not happen for ``non-trivial'' values of $\alpha$.}
Nevertheless, for many natural families of queries and choices of
$\alpha$, $\fnorminf(\qmat, \alpha/2)$ will be relatively stable to
small changes in $\alpha$, in which case our lower bound will be tight
up to this $O(\log(1/\alpha)^2\log |\queries|)$ factor. In contrast, existing characterizations for the central model~\cite{HardtT10, BhaskaraDKT12, NikolovTZ13, Nikolov15, BlasiokBNS19} lose a $\mathrm{poly}(1/\alpha)$ factor,  or else they lose a $\mathrm{polylog} |\uni|$ factor that is typically large.

\begin{remark} Our proof of Theorem~\ref{thm:main-local}, in fact, shows that the lower bound holds in the distributional setting where $\ds$ is sampled i.i.d.\ from an unknown distribution $\dist$, and the goal is to estimate the quantity $\query(\dist) = \ex{\elem \sim \dist}{\query(\elem)}$ for every query $\query \in \queries$ up to error at most $\alpha$.
\end{remark}

\begin{remark} Theorem~\ref{thm:main-local} crucially assumes that the
  error is bounded in the $\ell_\infty$ metric.  If we consider the
  less stringent $\ell_2^2$ error metric (appropriately scaled to
  reflect the error per query), then one can achieve sample complexity
  $O(\log |\uni| / \priv^2 \alpha^4)$ for any workload of
  queries~\cite{BlasiokBNS19}, which can be exponentially smaller than
  the lower bound we prove for $\ell_\infty$ error.  In many
  applications, such as releasing the PDF, CDF, or marginals of the
  data, the $\ell_\infty$ error metric is standard in the literature
  on these problems, and is more practical, since, for natural
  datasets, the weaker $\ell_2^2$ guarantee can be achieved by
  mechanisms that ignore the data.
\end{remark}

\noindent Using Theorem~\ref{thm:main-local}, we obtain new lower bounds for three well studied families of queries: 
\begin{enumerate}
	\item \emph{Threshold queries,} which are also known as range queries, and equivalent to computing the CDF of the data.
	\item \emph{Parity queries,} which capture the covariance and higher-order moments of the data.
	\item \emph{Marginal queries,} also known as conjunctions, which capture the marginal distribution on subsets of the attributes.
\end{enumerate}

\begin{corollary} [Thresholds / CDFs] \label{cor:threshold}
Let $\queries^{\textrm{cdf}}_{T}$ be the family of statistical queries over the domain $\uni = [T]$ that, for every $1 \leq t \leq T$, contains the statistical query $q_{t}(\elem) = \mathbb{I}\{ \elem \leq t \}$.  Then for every $T \in \N$ and $\priv,\alpha$ smaller than an absolute constant,
$$\sc_{ \priv, 0}^{\ell_\infty, \textrm{loc}}(\queries^{\textrm{cdf}}_{T}, \alpha) = \Omega\left( \log^2 T \right).
$$
\end{corollary}

We obtain this corollary by combining Theorem~\ref{thm:main-local} with results from~\cite{ForsterSSS03}.  Corollary~\ref{cor:threshold} should be compared to the upper bound of $O(\log^3 T)$ that can be obtained from the local analogue of the \emph{binary tree mechanism}~\cite{DworkNPR10,ChanSS11}.  Ours is the first lower bound to go beyond the easy $\Omega(\log T)$ lower bound for this problem, which follows easily via a so-called packing argument.

\begin{corollary} [Parities] \label{cor:parity}
	Let $\queries^{\textrm{parity}}_{d,w}$ be the family of statistical queries over the domain $\uni = \{\pm 1\}^d$ that, for every $S \subseteq [d], |S| \leq w$, contains the statistical query $q_{S}(\elem) = \prod_{j \in S} \elem_{j}$.  Then for every $k \leq d \in \N$ and $\priv,\alpha$ smaller than an absolute constant,
	$$\sc_{\priv, 0}^{\ell_\infty, \textrm{loc}}(\queries^{\textrm{parity}}_{d,w}, \alpha) = \Omega((d/w)^w).$$
\end{corollary}

Corollary~\ref{cor:parity} says that adding independent Gaussian noise to each query is optimal up to a $O(w \log(d/w))$ factor.
Using similar techniques,
one can also obtain a direct proof that gives a tight lower bound up to constant factors, even for the simpler problem of finding the subset $S$ of size at most $w$ that maximizes $q_{S}(\ds)$.

\begin{corollary} [Marginals] \label{cor:marginal}
Let $\queries^{\textrm{marginal}}_{d,w}$ be the family of statistical queries over the domain $\uni = \{0,1\}^d$ that, for every $S \subseteq [d], |S| \leq w$, contains the statistical query $q_{S}(\elem) =  \prod_{j \in S} \elem_{j}$.  Then for every $k \leq d \in \N$ and $\priv,\alpha$ smaller than an absolute constant,
$$\sc_{\priv, 0}^{\ell_\infty, \textrm{loc}}(\queries^{\textrm{marginal}}_{d,w}, \alpha) = (d/w)^{\Omega(\sqrt{w})}.$$
\end{corollary}

Marginal queries have been extremely well studied in differential privacy~\cite{BarakCDKMT07,KasiviswanathanRSU10,GuptaHRU11,HardtRS12,ThalerUV12,ChandrasekaranTUW14,DworkNT15}.  Corollary~\ref{cor:marginal} shows that a natural local analogue of the algorithm of~\cite{ThalerUV12} is optimal for answering marginal queries up to the hidden constant factor in the exponent.

\subsubsection{Agnostic Learning in the Local Model}
Theorem~\ref{thm:main-local} extends to characterizing \emph{agnostic PAC learning}~\cite{KearnsSS94} in the local model. In agnostic PAC learning, the dataset consists of labeled examples $\ds = ((\dsrow_1,y_1), \ldots, (\dsrow_\dsize, y_\dsize))$, where $\dsrow_i \in \uni$, and $y_i \in \{\pm 1\}$, and each pair $(\dsrow_i, y_i)$ is sampled independently from an unknown distribution $\dist$.  The goal is to find a concept $\concept: \uni \to \{\pm 1\}$ in a concept class $\concepts$ that approximately maximizes \(\E_{(\elem, y)\sim \dist}[\concept(\elem)y]\). 

The correlation of each concept $\concept$ with the labels in the data is a linear query, and one natural approach to agnostic PAC learning is to estimate all these linear queries, and output the concept that corresponds to the largest query value. Thus, we can apply the local approximate factorization mechanism to the family of queries $\concepts$ to obtain the same sample complexity upper bound in~\eqref{eq:lafm-ub}.  Interestingly, the proof of our lower bound in Theorem~\ref{thm:main-local} shows that the same lower bound also applies to this a priori easier problem of agnostic PAC learning, showing that the local approximate factorization mechanism gives an approximately optimal way to learn any concept class $\concepts$.

Prior results of Kasisiviswanathan et al.~\cite{KasiviswanathanLNRS08}
connecting learning algorithms in the local model with the SQ model,
together with characterizations of sample complexity in the SQ
model~\cite{BlumFJKMR94,Szorenyi09}, give  upper and lower bounds on
sample complexity of learning in the local model in terms of SQ
dimension. These results, however, are only tight up to polynomial
factors in the SQ dimension---which can be polynomial in
$|\concepts|$---whereas our results are sharper. We remark that, technically, the
results are not comparable, since the the characterization via the SQ
model holds for sequentially interactive, rather than non-interactive,
mechanisms. 

\subsubsection{Linear Queries in the Central Model}
Our second set of results quantitatively strengthens---and simplifies the proof of---the central model characterization of~\cite{NikolovTZ13}.  In contrast to the local model, the sample complexity of answering many natural workloads of linear queries exhibits two distinct regimes, depending on the desired accuracy.  For example, for a worst-case workload of linear queries, the sample complexity is at most
$$
\min\left\{ \frac{\log^{1/2} |\uni| \log |\queries|}{\priv \alpha^2}, \frac{|\queries|^{1/2}}{\priv \alpha} \right\}.
$$
Thus, the sample complexity behaves very differently when $\alpha$ goes below some critical value.  Our results concern this \emph{high-accuracy regime} where $\alpha$ is quite small.  In these results, we consider the $\ell_2^2$ error (scaled to be directly comparable to the $\ell_\infty$ error), which is
$$
\err^{\ell_2^2}(\mech,\queries, n) = \max_{\ds \in \uni^n} \ex{\mech}{\tfrac{1}{|\queries|}\| \mech(\ds) - \queries(\ds) \|_2^2}^{1/2}
$$
with the related quantities defined analogously.  Notice that we have scaled the $\ell_2^2$ error so that $\err^{\ell_2^2}(\mech, \queries, n) \leq \err^{\ell_\infty}(\mech,\queries,n)$.  For $\ell_2^2$ error, the natural factorization norm that describes the error of the factorization mechanisms is
$$
\fnorm(\qmat) = \set{ \tfrac{1}{|\queries|^{1/2}} \| R \|_{F} \| A \|_{1 \to 2} : \qmat = RA},
$$
where $\|R\|_F = \sqrt{\sum_{i,j} R_{i,j}^2}$ is the Frobenius norm of $R$.

In this high-accuracy regime, a combination of~\cite{NikolovTZ13}
and~\cite{NikolovT15} (see also the thesis~\cite{nikolov-thesis}) shows that, for every workload of linear queries, there is some $\alpha^*$ such that
\begin{align*}
\forall \alpha \leq \alpha^* \ \ \Omega( \log^{-1} |\queries|) \cdot \frac{\fnorm(\qmat)}{\priv \alpha} \leq \sc_{\priv, \privd}^{\ell_2^2}(\queries, \alpha) \le O(1) \cdot \frac{\fnorm(\qmat)}{\priv \alpha}.
\end{align*}
Note that the upper and lower bound differ by a factor of $O(\log
|\queries|)$.  The upper bound above is precisely what is given by the factorization mechanism, but the lower bound is smaller than the relevant factorization norm by an $\Omega(\log |\queries|)$ factor.  Our next theorem removes this log factor from the lower bound, and thus gives a characterization up to $O(1)$ for $\ell_2^2$.
\begin{theorem} \label{thm:main-central}
Let $\priv, \privd > 0$ be smaller than some absolute constants and let $\queries$ be a workload of linear queries with workload matrix $\qmat$.  There exists some $\alpha^* > 0$ such that for every $\alpha \leq \alpha^*$,
\begin{equation*}
 \sc_{\priv, \privd}^{\ell_2^2}(\queries, \alpha) = \Omega\left( \frac{\fnorm(\qmat) }{ \priv \alpha }\right).  
\end{equation*}
\end{theorem}

In addition to being sharper, our proof of Theorem~\ref{thm:main-central} is dramatically simpler than the lower bounds in~\cite{NikolovTZ13,NikolovT15}.  

\begin{remark} By a trivial reduction, Theorem~\ref{thm:main-local},
  in fact, gives
lower bounds for the distributional setting where $\ds$ is sampled i.i.d.\ from an unknown distribution $\dist$, and the goal is to estimate the quantity $\query(\dist) = \ex{\elem \sim \dist}{\query(\elem)}$ for every query $\query \in \queries$ up to error at most $\alpha$.
\end{remark}

\mypar{Data-Independent Mechanisms}  Along the way, we prove a simple result that this sample complexity bound holds for \emph{every} choice of $\alpha$, provided we restrict attention to \emph{data-independent mechanisms.}  These mechanisms can be written in the form $\mech(\ds) = \queries(\ds) + Z/n$ for some fixed random variable $Z$ that depends only on $\queries$ and not on the data.

For such mechanisms we show that the sample complexity is always $\Omega(\fnorm(\qmat)/\priv \alpha)$, regardless of $\alpha$.\footnote{Technically, we require $\alpha \leq \| \qmat \|_{1 \to \infty}$, but in nearly all applications of interest $\| \qmat \|_{1 \to \infty} = 1$.}

Data-independent mechanisms are interesting on their own, since the fact that we add noise from a known distribution makes them simpler to implement, and also means that we can give precise confidence intervals on the error of the mechanism.  One application of our lower bound for data-independent mechanisms is an $\Omega(\log T)$ lower bound on the sample complexity of any mechanism for answering threshold queries over $[T]$ in $\ell_2^2$ error, which matches the data-independent binary tree mechanism.


\subsection{Techniques}

Below we give a brief overview of the techniques used to prove Theorems~\ref{thm:main-local}~and~\ref{thm:main-central}.  

\mypar{Lower bound in the local model} As mentioned above, Theorem~\ref{thm:main-local} is proved in the distributional setting, where the dataset $\ds$ consists of $\dsize$ i.i.d.\ samples from some distribution $\dist$, and the goal is to estimate the expectation of each query $\query \in \queries$ on $\dist$.  Our approach is to design two families of hard distributions $\{\dista_1, \ldots, \dista_\qsize\}$ and $\{\distb_1, \ldots, \distb_\qsize\}$ with the following properties:  first, any locally differentially private mechanism requires many samples to distinguish these two families;  second, the two families give very different answers to the queries.

To show that the distributions are hard to distinguish, we prove an upper bound on the KL-divergence between: (1) the transcript of a private mechanism in the local model when run on $\dsize$ samples from a random distribution in $\{\dista_1, \ldots, \dista_\qsize\}$, and (2) the same, but for a random distribution in $\{\distb_1, \ldots, \distb_\qsize\}$. Intuitively, the bound shows that the KL-divergence between transcripts is small when no bounded test function can simultaneously distinguish between $\dista_v$ and $\distb_v$ on average over a random choice of $v \in [\qsize]$.  This bound is a slight extension of a similar bound from~\cite{DJW}.  In particular, the upper bound on the KL-divergence is in terms of the $\infty \to 2$ operator norm of a matrix $M$ derived from the two families of distributions.

Thus, what remains is to find families distributions $\{\dista_1, \ldots, \dista_\qsize\}$ and $\{\distb_1, \ldots, \distb_\qsize\}$, for which the $\infty \to 2$ operator norm of $M$ is small, but the expectations of the queries in $\queries$ are sufficiently different on the two families.  
Recall that our goal is to prove a lower bound in terms of the approximate norm $\fnorminf(\qmat, \alpha)$, where $\qmat$ is the workload matrix. Since $\fnorminf(\qmat, \alpha)$ is the value of a convex minimization problem, it admits a dual characterization, showing that $\fnorminf(\qmat, \alpha)$ is equal to the value of a maximization problem over matrices $U$. We take an optimal dual solution $U$, and use it to derive distributions $\{\dista_1, \ldots, \dista_\qsize\}$ and $\{\distb_1, \ldots, \distb_\qsize\}$. The objective function of the dual problem guarantees that these distributions are such that the expectation of any query $\query \in \queries$ on any $\dista_v$ is small, yet the expectation of the query $\query_v$ on $\distb_v$ is large. Moreover, the dual objective, together with classical arguments in functional analysis, also guarantees an upper bound on the $\infty \to 2$ norm of the appropriate matrix $M$, giving us both ingredients for our lower bound.

\mypar{Lower bound in the central model}
The main ingredient of the proof of Theorem~\ref{thm:main-central} is a lower bound of $\Omega(\fnorm(\qmat)/\priv \alpha)$ on the sample complexity of data-independent mechanisms. Recall that a mechanism $\mech$ is data-independent if $\mech(\ds) = \queries(\ds) + \frac{1}{n} Z$ for a random variable $Z \in \R^{\queries}$. Our key observation is that, if $\Sigma$ is the covariance matrix of $Z$, then the mechanism $$\queries(\ds) + \frac{O(\log(1/\delta))}{n} \cdot \mathcal{N}(0,\Sigma)$$ that uses Gaussian noise in place of $Z$ is also $(\priv,\privd)$-differentially private.  Moreover, the $\ell_2^2$ error of $\mech$ is equal to $\tr(\Sigma)/|\queries|^{1/2}$, so, up to a factor of $O(\log(1/\delta))$, the optimal data-independent mechanism with respect to $\ell_2^2$-error can be assumed to use correlated Gaussian noise. It is easy to see that the class of all such mechanism is equivalent to the class of all factorization mechanisms, and, hence, the optimal achievable $\ell_2^2$-error is $O(\fnorm(\qmat)/ \priv n)$. 

To give a lower bound for arbitrary mechanisms in the high-accuracy regime, we use a clever transformation from~\cite{BhaskaraDKT12} that turns a data-dependent mechanisms that is accurate for large datasets into a data-independent mechanism.

\section{Preliminaries} \label{sec:prelims}

In this section we recount basic notation and definitions used
throughout the paper.

\subsection{Norms}

For a set $\mathcal{S}$, the $\ell_1$, $\ell_2$ and $\ell_\infty$ norms
on $\R^\mathcal{S}$ are given respectively by
\[
    \| a \|_1 = \sum_{v \in \mathcal{S}} |a_v|,
    \quad
    \| a \|_2 = \sqrt{ \sum_{v \in \mathcal{S}} (a_v)^2 },
    \quad
    \| a \|_\infty = \max_{v \in \mathcal{S}} |a_v|.
\]

Given a probability distribution $\pi$ on $\mathcal{S}$,
we consider the norms $\|\cdot\|_{L_1(\pi)}$ and $\|\cdot\|_{L_2(\pi)}$
on $\R^\mathcal{S}$, given by
\[
    \| a \|_{L_1(\pi)}
    = \sum_{v \in \mathcal{S}} \pi(v) |a_v|,
    \quad 
    \| a \|_{L_2(\pi)}
    = \sqrt{ \sum_{v \in \mathcal{S}} \pi(v) (a_v)^2}.
\]

We also take advantage of a number of matrix norms.
For norms $\|\cdot\|_\zeta$ and $\|\cdot\|_\xi$
on $\R^{\mathcal{S}}$ and $\R^{\mathcal{S}'}$ respectively,
we consider the \emph{matrix operator norm} of
$M \in \R^{\mathcal{S} \times \mathcal{S}'}$
given by
\[
    \| M \|_{\zeta \to \xi} = \max_{x \in \R^\mathcal{S} \setminus \{0\}} \frac{\| Mx \|_\xi}{\|x\|_\zeta}.
\]
For the special case of $\|M\|_{\ell_s \to \ell_t}$,
we will simply write $\|M\|_{s \to t}$.
Of particular importance are
$\|M\|_{1 \to \infty}$ which corresponds to the largest entries of $M$,
$\|M\|_{1 \rightarrow 2}$,
which corresponds to the maximum $\ell_2$-norm of a column of $M$,
and
$\|M\|_{2 \rightarrow \infty}$,
which corresponds to the maximum $\ell_2$-norm of a row of $M$.

The \emph{inner product} of two matrices $M$ and $N$ in
$\R^{\mathcal{S} \times \mathcal{S}'}$ is defined by
$M\bullet N = \tr(M^\top N) = \sum_{u \in \mathcal{S}, v\in
  \mathcal{S}'} m_{u,v} n_{u,v}$.
The \emph{Frobenius norm} of
$M \in \R^{\mathcal{S} \times \mathcal{S}'}$
is given by $\|M\|_F = \sqrt{M\bullet M}$.

Lastly, the \emph{factorization norms} $\fnorm$ and $\fnorminf$
central to this work are given for
$M \in \R^{\mathcal{S} \times \mathcal{S}'}$ by
\begin{align*}
    &\fnorm(M) = \min\left\{\tfrac{1}{|\mathcal{S}|^{1/2}} \|R\|_{F}\|A\|_{1 \to 2} : RA = M \right\}, \\
    &\fnorminf(M) = \min\{\|R\|_{2 \to \infty}\|A\|_{1 \to 2} : RA = M\}.
\end{align*}

\subsection{Differential privacy}

Let $\uni$ denote the \emph{data universe}.
A generic element from $\uni$ will be denoted by $\elem$.
We consider \emph{datasets} of the form
$\ds = (\elem_1,\dots,\elem_\dsize) \in \uni^n$,
each of which is identified with its \emph{histogram}
$\hist \in \Z_{\ge 0}^{\uni}$ where, for every $\elem \in \uni$, $\hist_{\elem} = | \set{i : \dsrow_i = \elem} |$,
so that $\| \hist \|_1 = \dsize$.
To refer to a dataset, we use $\ds$ and $h$ interchangeably.
A pair of datasets $\ds = (\elem_1,\dots,\elem_i,\dots,\elem_\dsize)$
and $\ds' = (\elem_1,\dots,\elem_i',\dots,\elem_\dsize)$ are called \emph{adjacent} if $X'$ is obtained from $X$ by replacing
an element $\elem_i$ of $\ds$ with a new universe element $\elem_i'$.

For parameters $\priv,\privd>0$,
an \emph{$(\priv,\privd)$-differentially private mechanism}~\cite{DworkMNS06} (or
$(\priv,\privd)-DP$ for short)
is a randomized function $\mech:\uni^\dsize \to \outspace$
which, for all adjacent datasets $\ds$ and $\ds'$,
for all outcomes $S \subseteq \outspace$, satisfies
\[
    \Pr_\mech[\mech(X) \in S] \le e^\priv\Pr_\mech[\mech(X') \in S] + \privd.
\]
A mechanism which is $(\priv,0)$-differentially private 
will be referred to as being simply \emph{$\priv$-differentially
  private} (or $\priv$-DP for short).

Of special interest are $(\priv,\privd)$-differentially private mechanisms $\mech_i:\uni \to \loutspace$
which take a singleton dataset $X = \{x\}$ as input.
These are referred to as \emph{local randomizers}.
A sequence of $(\priv,\privd)$-differentially private local randomizers $\mech_1,\dots,\mech_\dsize$
together with a \emph{post-processing function}
$\post:\loutspace^\dsize \to \outspace$
specify a \emph{(non-interactive) locally
  $(\priv,\privd)$-differentially private mechanism}
$\mech:\uni^\dsize \to \outspace$~\cite{EvfimievskiGS03,DworkMNS06, KasiviswanathanLNRS08}. In short, we say that such
mechanisms are $(\priv, \privd)$-LDP, or $\priv$-LDP when $\privd=0$.
When the local mechanism $\mech$ is applied to a dataset $\ds$,
we refer to
\[
    \trans(X) = (\mech_1(\elem_1),\dots,\mech_\dsize(\elem_\dsize))
\]
as the \emph{transcript} of the mechanism.
Then the output of the mechanism is given by
\(
    \mech(X) = \post(\trans(X)).
\)

\subsection{Linear queries}

A \emph{linear query} is specified by a bounded function
$q: \uni \to \R$.
Abusing notation slightly, its answer on a dataset $\ds$ is given by
$\query(\ds) = \frac{1}{\dsize} \sum_{i=1}^{\dsize} \query(\dsrow_i)$.
We also extend this notation to distributions: if $\dist$ is a
distribution on $\uni$, then we write $\query(\dist)$ for
 $\ex{\elem \sim \dist}{\query(\elem)}$.
A \emph{workload} is a set of linear queries
$\queries = \set{\query_1,\dots,\query_\qsize}$,
and
$\queries(\ds) = (\query_1(\ds),\dots,\query_\qsize(\ds))$
is used to denote their answers. The answers on a distribution $\dist$
on $\uni$ are denoted by 
$\queries(\dist) = (\query_1(\dist),\dots,\query_\qsize(\dist))$.
We will often represent $\queries$ by its \emph{workload matrix} $W \in \R^{\queries \times \uni}$ with entries $w_{q,x} = q(x)$.
In this notation, the answers to the queries are given by
$\frac{1}{n} \qmat \hist$.
We will often use $\queries$ and $\qmat$ interchangeably.

\subsection{Error and sample complexity}

The \emph{$\ell_\infty$ and $\ell_2^2$-error} of a mechanism $\mech$,
which takes a dataset of size $\dsize$,
on the query workload $\queries$ are given by
\begin{align*}
    \err^{\ell_\infty}(\mech,\queries, n) &= \max_{\ds \in \uni^n}
    \ex{\mech}{\| \mech(\ds) - \queries(\ds) \|_\infty},\\
    \err^{\ell_2^2}(\mech,\queries, n) &= \max_{\ds \in \uni^n} 
    \ex{\mech}{\tfrac{1}{|\queries|}\| \mech(\ds) - \queries(\ds) \|_2^2}^{1/2}.
\end{align*}
We can then define the \emph{sample complexity} of a mechanism $\mech$
for a given $\ell_\infty$ error $\alpha$ by 
\[
\sc^{\ell_\infty}_{\eps,\delta}(\mech,\queries,\alpha) = \min\{ n : \err^{\ell_\infty}(\mech,\queries, n) \leq \alpha \}.
\]
The sample complexity with respect to $\ell_2^2$ error $\sc^{\ell_2^2}_{\eps,\delta}(\queries,\alpha)$ is defined analogously.

Having defined error and sample complexity for a fixed mechanism, we
can define the optimal error and sample complexity by
\begin{align*}
    \err^{\ell_\infty}_{\priv, \privd}(\queries, n) &=
    \min_{\textrm{$\mech$ is $(\priv,\privd)$-DP}}
    \err^{\ell_\infty}(\mech,\queries,n),\\
    \sc^{\ell_\infty}_{\priv, \privd}(\queries, \alpha) &=
    \min_{\textrm{$\mech$ is $(\priv,\privd)$-DP}}
    \sc^{\ell_\infty}(\mech,\queries,n).
\end{align*}
The analogous quantities $\err^{\ell_2^2}_{\priv, \privd}(\queries,
n)$ and $\sc^{\ell_2^2}_{\priv, \privd}(\queries, \alpha)$ for
$\ell_2^2$-error are defined similarly. The optimal error and sample
complexity for the local model are denoted 
$\err^{\ell_\infty,\textrm{loc}}_{\priv, \privd}(\queries,n)$ and 
$\sc^{\ell_\infty,\textrm{loc}}_{\priv, \privd}(\queries, \alpha)$,
and are defined in the same way but with the minimum taken over
$(\priv,\privd)$-LDP mechanisms.

\subsection{Factorization Mechanisms}

The Gaussian mechanism~\cite{DinurN03,DworkN04,DworkMNS06} is defined as
\begin{align*}
\mech_{\textrm{Gauss}}(\qmat, \hist) = \frac{1}{n} \qmat \hist + Z,
&&Z\sim\mathcal{N}\left( 0 , \left(\frac{\gauss \| \qmat \|_{1 \to 2}}{n}\right)^2 \cdot I  \right),
\end{align*}
where $\gauss= O(\sqrt{\log(1/\delta)}/\eps)$ depends only on the
privacy parameters. Given a factorization $\qmat = RA$, we consider
the mechanism 
\begin{align*}
\mech_{R,A}(\hist)
= R \ \mech_{\textrm{Gauss}}(\qmat, \hist)
&={} \frac{1}{n} \qmat \hist + Z,
&&Z \sim\mathcal{N}\left( 0 , \left(\frac{\gauss \| A \|_{1 \to 2}}{n}\right)^2  \cdot R R^\top  \right),
\end{align*}
and, utilizing Gaussian tail bounds, one can show that the error is 
$$
\err^{\ell_\infty}(\mech_{R,A}, \queries, n) = O\left( \frac{\| R \|_{2 \to \infty} \| A \|_{1 \to 2} \sqrt{\log(1/\delta) \log |\queries|}}{ \eps n}\right).
$$
We define the \emph{factorization mechanism} $\mech_{\fnorminf}$ to be the mechanism that chooses $R,A$ to minimize this expression, and its error is proportional to the \emph{factorization norm}
$$
\fnorminf(\qmat) = \min \{ \| R \|_{2 \to \infty} \| A \|_{1 \to 2} : \qmat = RA \}.
$$
The sample complexity of this mechanism is thus 
$$
\sczinf(\mech_{\fnorminf},\queries,\alpha) = O\left(\frac{\fnorminf(\qmat) \sqrt{\log(1/\delta) \log|\queries|}}{\alpha}\right).
$$
This mechanism is implicit in~\cite{NikolovTZ13}, and is stated in
this form in~\cite{nikolov-thesis}.

Analogously, we can show that 
\[
\err^{\ell_2^2}(\mech_{R,A}, \queries, n) = O\left( \frac{|\queries|^{-1/2}\| R \|_{F} \| A \|_{1 \to 2} \sqrt{\log(1/\delta)}}{ \eps n}\right).
\]
Optimizing this error bound over the choice of $R$ and $A$ gives error
proportional to the factorization norm
\[
\fnorm(\qmat) = \min \{|\queries|^{-1/2} \| R \|_{F} \| A \|_{1 \to 2} : \qmat = RA \},
\]
and the mechanism $\mech_{\fnorm}$ that runs $\mech_{R, A}$ with the
$R$ and $A$ achieving $\fnorm(\qmat)$ has sample complexity 
\[
\sc^{\ell_2^2}_{\eps,\delta}(\queries,\alpha) =  O\left(\frac{\fnorm(\qmat) \sqrt{\log(1/\delta) }}{\alpha}\right).
\]
This factorization mechanism is equivalent to the Gaussian noise
matrix mechanism in~\cite{LiHRMM10}.

\section{Non-Interactive Local DP: Linear Queries}

In this section we give details about our results for answering linear
queries in the local model. We first present the local approximate
factorization mechanism. Then we give an information theoretic lemma
that bounds the KL-divergence between the transcripts of mechanisms in
the local model on inputs drawn from mixtures of product
distributions. We then use a dual formulation of the approximate
$\fnorminf$ norm to construct distributions to use with the
information theoretic lemma in order to prove the lower bound in
Theorem~\ref{thm:main-local}. 

\subsection{Approximate Factorization}

Here we give details of the approximate factorization mechanism, which
was sketched in the introduction. Recall that the approximate
$\fnorminf$ norm is defined by 
\[
\fnorminf(\qmat,\alpha) = \min \{
\fnorminf(\widetilde\qmat) : \| \qmat - \widetilde\qmat \|_{1 \to
  \infty} \leq \alpha/2 \},
\]
where $\fnorminf(\widetilde\qmat) = \min \{ \| R \|_{2 \to \infty} \| A
\|_{1 \to 2} : \qmat = RA \}$. Matrices $\widetilde\qmat$, $R$, and
$A$ achieving the minimum to any degree of accuracy can be computed in
polynomial time via semidefinite programming, as shown
in~\cite{LinialS09}. Our main positive result shows that the sample
complexity of the corresponding approximate factorization mechanism is
bounded above by the approximate $\fnorminf$ norm. As sketched in the
introduction, this can be achieved via a local version of the Gaussian
noise mechanism, which can then be transformed into a purely private
mechanism using the results of~\cite{BunNS18}. This gives, however, a
slightly suboptimal bound, and, instead, we use the local randomizer
from~\cite{BlasiokBNS19}, which is a variant of a local randomizer
from~\cite{DJW}. The relevant properties of this local randomizer are
captured by the next lemma. We recall that a random variable $Z$ over
$\R$ is $\sigma$-subgaussian if $\E \exp(Z^2/\sigma^2) \le 2$, and a
random variable $Z$ over $\R^d$ is $\sigma$-subgaussian if
$\theta^\top Z$ is $\sigma$-subgaussian for every vector $\theta$ such
that $\|\theta\|_2 = 1$.
\begin{lemma}[\cite{BlasiokBNS19}]
  \label{lm:local-release}
  There exists an $\priv$-DP mechanism $\mech$
  which takes as input a single datapoint $x \in \R^\inner$ such that $\|x\|_2 \le
  1$, and outputs a random $Y_x := \mech(x) \in \R^\inner$  such that 
  \begin{enumerate}
  \item $Y_x$ can be sampled in time polynomial in $d$ on input $x$,
  \item $\ex{}{Y_x} = x$,
  \item $Y_x - x$ is $\sigma$-subgaussian with $\sigma =
    O(\priv^{-1})$.
  \end{enumerate}
\end{lemma}

Given this local randomizer, and approximate factorizations, we are ready to prove our upper bound.

\begin{theorem}[Approximate Factorization Mechanism]\label{thm:local-apxfact}
There exists an $\priv$-LDP mechanism $\mech_{\fnorminf, \alpha}^{\textrm{loc}}$ such that, for any $\qsize$ statistical queries $\queries$ with workload matrix $\qmat$, we have 
\[
\sczinf(\mech_{\fnorminf, \alpha}^{\textrm{loc}},\queries,\alpha) =
O\left(\frac{\fnorminf(\qmat, \alpha/2)^2 \log \qsize}{ \priv^2 \alpha^2}\right),
\]
and the mechanism runs in time polynomial in $\dsize$, $\qsize$, and
$|\uni|$.
\end{theorem}
\begin{proof}
  Let $\widetilde{\qmat}$, $R$ and $A$ be such that
  $\|\widetilde{\qmat} - \qmat\|_{1 \to \infty} \le \frac{\alpha}{2}$,
  $\widetilde{\qmat} = RA$, 
  $\|R\|_{2 \to \infty} = \fnorminf(\qmat, \alpha/2)$, 
  and $\|A\|_{1 \to 2}= 1$. 
  Such matrices always exist, because if $R$ and $A$ achieve
  $\fnorminf(\qmat, \alpha/2)$, so do $tR$ and $A/t$, and we can
  choose $t$ so that $\|A\|_{1 \to 2}= 1$. Moreover, $R$ and
  $A$ can be computed in polynomial time via semidefinite programming,
  as noted above.

  Let $\ds = (\dsrow_1,\dots,\dsrow_n)$ be the dataset, and let
  $\hist$ be its histogram.  Each agent $i$ holds a data point
  $\dsrow_i$, and the histogram for the single point dataset
  containing $\dsrow_i$ is $\hist_i := e_{\dsrow_i}$, i.e.~the standard
  basis vector of $\R^\uni$ corresponding to $\dsrow_i$. Agent $i$
  releases \( \mech_i(\dsrow_i) = Y_{A\hist_i}, \) where $Y_{A\hist_i}$ is as
  defined in Lemma~\ref{lm:local-release}.  Then $Y = \frac{1}{n}
  \sum_{i=1}^n \mech_i(\dsrow_i)$ has expectation 
  \[
  \frac1n \sum_{i = 1}^n A\hist_i = \frac1n A \sum_{i = 1}^n \hist_i 
  = \frac1n A\hist.
  \]
  Moreover, since $Y - \frac1n A\hist = \frac1n \sum_{i =
    1}^n{(Y_{A\hist_i} - A\hist_i)}$ is the average of $n$ independent
  $\sigma$-subgaussian random variables, where $\sigma =
  O(\priv^{-1})$ is as in Lemma~\ref{lm:local-release}, $Y - \frac1n A\hist$ is
  $O\left(\frac{\sigma}{\sqrt{n}}\right)$-subgaussian (see
  e.g.~\cite[Proposition 2.6.1.]{Vershynin}). 

  Post-processing $Y$ by our reconstruction matrix $R$ gives the
  output of our mechanism, namely $\mech(\ds) = RY$. Then $\E[RY] =
  RA\hist = \frac1n\widetilde{\qmat}\hist$, and we have
  \begin{equation}\label{eq:error-expect}
    \left\|\E[RY] - \frac1n \qmat h\right\|_\infty = 
  \frac1n\|(\widetilde{\qmat} - \qmat)\hist\|_\infty 
  \le \frac1n\|\widetilde{\qmat} - \qmat\|_{1\to\infty} \|\hist\|_1
  \le \frac{\alpha}{2}.
  \end{equation}
  Every coordinate of $R(Y - \frac1n A\hist) = RY -
  \E[RY]$ is the inner product of $Y - \frac1n
  A\hist$ and a row of $R$, the latter having $\ell_2$ norm at most
  $\|R\|_{2 \to \infty}$. Since $Y - \frac1n A\hist$ is
  $O\left(\frac{\sigma}{\sqrt{n}}\right)$-subgaussian, every
  coordinate $(RY - \E[RY])_{\query}$ for every
  $\query \in \queries$, is
  $O\left(\frac{\sigma\|R\|_{2\to\infty}}{\sqrt{n}}\right)$-subgaussian. It
  is then a standard fact (see e.g.~\cite[Exercise~2.5.10]{Vershynin})
  that
  \[
  \E\|RY - \E[RY]\|_\infty
  = O\left(\frac{\sigma\|R\|_{2\to\infty}\sqrt{\log \qsize}}{\sqrt{n}}\right)
  = O\left(\frac{\fnorminf(\qmat, \alpha/2)\sqrt{\log \qsize}}{\priv\sqrt{n}}\right).
  \]
  Combining with \eqref{eq:error-expect}, and applying the triangle
  inequality, we get
  \begin{align*}
  \E \|\mech(\ds) - \qmat\hist\|_\infty
  &= \E\|RY - \qmat\hist\|_\infty \\
  &\le \E\|RY - \E[RY]\|_\infty +   \|\E[RY] - \tfrac1n \qmat h\|_\infty\\
  &= 
  \frac{\alpha}{2} + O\left(\frac{\fnorminf(\qmat, \alpha/2)\sqrt{\log \qsize}}{\priv\sqrt{n}}\right).
  \end{align*}
  The proof is completed by setting $n$ so that the second term is at
  most $\frac{\alpha}{2}$.
\end{proof}

\subsection{Bounding KL-Divergence}

Our lower bound will rely on the construction, based on a workload
$\queries$, of families $\{\dista_1,\dots,\dista_\qsize\}$ and
$\{\distb_1,\dots,\distb_\qsize\}$ of distributions on $\uni$.
Together with these, we consider a distribution $\pi$ over $[\qsize]$.
For any $v \in [\qsize]$, let $\dista_v^\dsize$ be the product
distribution induced by sampling $n$ times independently from
$\dista_v$, and let $\dista_\pi^\dsize$ be the mixture $\sum_{v =
  1}^{\qsize}\pi(v) \dista^\dsize_v$. Define $\distb_v^\dsize$ and
$\distb_\pi^\dsize$ analogously. Note that $\dista_\pi^\dsize$ and
$\distb_\pi^\dsize$ are \emph{not} product distributions, but mixtures
of such distributions.  For a mechanism $\mech$ in the local model,
and a probability distribution $\nu$ on $\uni^\dsize$, we use
$\trans(\nu)$ to denote the distribution on random transcripts
$\trans(\ds)$ when $\ds$ is sampled from $\nu$. Similarly, if $\nu$ is
a distribution on $\uni$, we use the notation $\mech_i(\nu)$ for the
distribution of $\mech_i(\dsrow)$, when $\dsrow$ is sampled from
$\nu$.

We approach the task of showing that
$\dista_1,\dots,\dista_\qsize$ and $\distb_1,\dots,\distb_\qsize$
are ``hard'' distributions on which to evaluate $\queries$
 in two steps.
On the one hand, we wish to argue that
being able to estimate $\queries$ on the distributions
$\dista_1,\dots,\dista_\qsize$ and $\distb_1,\dots,\distb_\qsize$
enables us to distinguish between $\dista_\pi^\dsize$ and $\distb_\pi^\dsize$.
On the other hand,
we show a lower bound on the number of samples required
for a locally private mechanism
to distinguish between $\dista_\pi^\dsize$ and $\distb_\pi^\dsize$.
The second of these objectives will be met by way of
the following bound on KL-divergence. Similar bounds were proved
in~\cite{DJW,DuchiR18} when only one of the two distributions is a
mixture of products, and our proof is similar to the proof of
Theorem~2 in \cite{DuchiR18}.
Our proof is in Appendix~\ref{ap:kl-div}.

\begin{lemma}\label{lm:kl-div}
Let $\eps \in (0,1]$, and let $\mech$ be an $\priv$-DP mechanism in
the local model. Then
\[
\div(\trans(\dista_\pi^\dsize)\|\trans(\distb_\pi^\dsize)) \le O(n \priv^2)\cdot 
\max_{f \in \R^\uni: \|f\|_\infty \le 1} \ex{V \sim \pi}{\left(\ex{x \sim \dista_V}{f_x} - \ex{x \sim \distb_V}{f_x}\right)^2}.
\]
In matrix notation, define
the matrix $M \in \R^{[K] \times \uni}$ by $m_{v, x} = (\dista_v(x) - \distb_v(x))$.
Then
\[
    \div(\trans(\dista_\pi^\dsize)\|\trans(\distb_\pi^\dsize)) \le O(n \priv^2)\cdot 
    \|M\|_{\ell_\infty \to L_2(\pi)}^2.
\]
\end{lemma}

Being able to distinguish between $\trans(\dista_\pi^\dsize)$ and
$\trans(\distb_\pi^\dsize)$ with constant probability implies, by
Pinsker's inequality, that
$\div(\trans(\dista_\pi^\dsize)\|\trans(\distb_\pi^\dsize)) \ge \Omega(1)$.
Together with Lemma \ref{lm:kl-div}, this would imply
\[
    n =
    \Omega \left(
        \frac{ 1 }{ \priv^2 \cdot \|M\|_{\ell_\infty \to L_2(\pi)}^2 }
    \right).
\]
Hence, our goal will be to define our distributions
so that that $\|M\|_{\ell_\infty \to L_2(\pi)}^2$ is small
while still meeting the requirement that
estimating the queries $\queries$ allows us to distinguish
between $\dista_\pi^\dsize$ and $\distb_\pi^\dsize$.

\subsection{Duality for \texorpdfstring{$\gamma_2(\qmat,
    \alpha)$}{gamma2(W,alpha)} and the Dual Norm}

Recall that our goal is to prove a lower bound on the sample
complexity of mechanisms in the local model in terms of the
approximate $\fnorminf$ norm. We will do so via Lemma~\ref{lm:kl-div},
and the distributions $\{\dista_1,\dots,\dista_\qsize\}$ and
$\{\distb_1,\dots,\distb_\qsize\}$ will serve as a certificate of a
lower bound on the sample complexity. On the other hand, convex
duality can certify a lower bound on the approximate $\fnorminf$
norm. In the proof of our lower bounds, we will show that these dual
certificates for which the approximate $\fnorminf$ norm is large can be
turned into hard families of distributions to use in
Lemma~\ref{lm:kl-div}.

The key duality statement follows. This dual formulation for
the $\fnorminf(\qmat, \alpha)$ was also given
in~\cite{LinialS09} for the special case when $\qmat$ has entries in
$\{-1, +1\}$.\footnote{Note that in~\cite{LinialS09}, Linial and
  Shraibman use the notation $\fnorminf^\alpha(\qmat) =
  \inf\{\fnorminf(\widetilde{\qmat}): 1 \le \widetilde{w}_{ij} w_{ij}
  \le \alpha \ \ \forall i, j\}$. For sign matrices $\qmat$ this is
  equal to $\frac{\alpha+1}{2}\fnorminf(\qmat, (\alpha-1)/(\alpha+1))$ in our
  notation.} For completeness, here we rederive it in Appendix~\ref{ap:duality} by directly
applying the hyperplane separator theorem. 

\begin{lemma}\label{lm:duality}
For any $\qsize \times \usize$ matrix $\qmat$ and $\alpha$, 
\[
\fnorminf(\qmat, \alpha) 
= \max
\frac{\qmat\bullet U - \alpha\|U\|_1}{\fnorminf^\ast(U)},
\]
where the max is over $\qsize \times \usize$ matrices $U\neq 0$,
and $\fnorminf^*$ is the dual norm to $\fnorminf$, given by 
\[
\fnorminf^\ast(U) = \max\{U\bullet V: \fnorminf(V) \le 1\}
= \max \sum_{i=1}^\qsize\sum_{j = 1}^\usize u_{i,j} y_i^\top z_j,
\]
where $a_1, \ldots, a_\qsize$ and $b_1, \ldots, b_\usize$ range over
vectors with unit $\ell_2$ norm in $\R^{\qsize + \usize}$.
\end{lemma}

The expression 
\[
\fnorminf^\ast(U) = \max \sum_{i=1}^\qsize\sum_{j = 1}^\usize u_{i,j} a_i^\top b_j,
\]
with the max over unit vectors $a_1, \ldots a_\qsize$ and $b_1,
\ldots, b_\usize$ can be easily formulated as a semidefinite program,
and, in fact, is exactly the semidefinite program that appears in
Grothendieck's inequality
(see,e.g.,\cite{KhotNaor-Grothendieck,Pisier-Grothendieck}). It is
straightforward to check (just take all the $a_i$ and $b_j$ co-linear)
that
\begin{equation}\label{eq:gammastar-opnorm}
\fnorminf^\ast(U) \ge \max\{y^\top U z: y \in \{-1, 1\}^m, z \in \{-1, 1\}^N\}
= \|U\|_{\infty \to 1}.
\end{equation}
Moreover, Grothendieck showed that this inequality is always tight up
to a universal constant~\cite{Grothendieck53}, although this fact will
not be used here. Instead, we will need the following lemma, which can
be derived from SDP duality, and is also due to Grothendieck. For a
proof using the Hahn-Banach theorem, see~\cite{Pisier-Grothendieck}.

\begin{lemma}[\cite{Grothendieck53}]\label{lm:gammastar-dual}
  For any $\qsize\times\usize$ matrix $U$, $\fnorminf^\ast(U) \le t$
  if and only if there exist diagonal matrices $P \in \R^{\qsize
    \times \qsize}$ and $Q \in \R^{\usize \times \usize}$, and a
  matrix $\widetilde{U} \in \R^{\qsize\times\usize}$ such that $\tr(P^2) = \tr(Q^2) = 1$,  $U =  P\widetilde{U} Q$, and $\|\widetilde{U}\|_{2 \to 2} \le t$.
\end{lemma}


By \eqref{eq:gammastar-opnorm}, the $\fnorminf^\ast(\cdot)$ norm is an upper
bound on the $\|\cdot\|_{\infty \to 1}$ norm. We use
Lemma~\ref{lm:gammastar-dual} to show a similar upper bound on the
$\|\cdot\|_{\infty \to 2}$, which allows projecting out some of the
rows of the matrix, but is quantitatively stronger. The reason we are
interested in the $\|\cdot \|_{\infty \to 2}$ norm is that this is the
norm that appears in the statement of Lemma~\ref{lm:kl-div}.

\begin{lemma}\label{lm:gammastar-inftyto2}
    For any matrix $U \in \R^{\qsize \times \usize}$,
    there exists a set $S \subseteq [\qsize]$ of size $|S| \ge \frac{\qsize}{2}$ such that
    \(
        \sqrt{\frac{\qsize}{2}} \|\Pi_SU\|_{\infty \to 2}
        \le \gamma_2^\ast(U),
    \)
    where $\Pi_S$ is the projection onto the subspace $\R^S$.
\end{lemma}

The next lemma slightly strengthens Lemma~\ref{lm:gammastar-inftyto2}
to allow for weights on the rows of the matrix. This is the key fact
about the $\fnorminf^*$ norm that we need for our lower bounds. 

\begin{lemma}\label{lm:gammastar-inftyto2-nonuniform}
  Let $U$ and $M$ be $\qsize \times \usize$ matrices, and let $\pi$ be
  a probability distribution on $[\qsize]$ where, for any $i \in
  [\qsize], j \in [\usize]$, we have $u_{i,j} = \pi(i)m_{i,j}$. Then
  there exists a probability distribution
  $\widehat{\pi}$ on $[\qsize]$,
  with support contained in the support of $\pi$,
  such that
    \(
        \| M \|_{\ell_\infty \to L_2(\widehat{\pi})}
        \le 4\fnorminf^\ast(U).
    \)
\end{lemma}

Lemmas \ref{lm:gammastar-inftyto2} and \ref{lm:gammastar-inftyto2-nonuniform}
are proved in Appendix \ref{ap:duality}.

\subsection{Symmetrization}

For our lower bound,
it will be convenient to narrow our attention
to the following restricted class of `symmetric' query workloads.

\begin{definition}\label{def:symmetric}
    Let $\queries$ be a workload of statistical queries
    with workload matrix $\qmat \in \R^{\queries \times \uni}$.
    Suppose there exists a partition of $\uni$
    into sets $\uni^+$ and $\uni^-$, $|\uni^+| = |\uni^-|$,
    where each element $\elem$ of $\uni^+$ is identified
    with a distinct element of $\uni^-$, denoted $-\elem$,
    such that,
    for all $\query \in \queries$,
    for all $\elem \in \uni$,
    $\query(-\elem) = -\query(\elem)$.
    In other words,
    $\qmat$ can be expressed as $(\qmat^+,\qmat^-)$,
    where $\qmat^+ \in \R^{\queries \times \uni^+}$
    and $\qmat^- \in \R^{\queries \times \uni^-}$
    are the restrictions of $\qmat$ to
    $\queries \times \uni^+$ and $\queries \times \uni^-$
    respectively,
    with each entry
    $\qent_{\query,\elem}^+$ of $\qmat^+$
    and the corresponding entry $\qent_{q,-\elem}^-$ of $\qmat^-$
    satisfying $\qent_{\query,\elem}^- = - \qent_{\query,-\elem}^+$.
    Also write $\queries^+$ to denote the collection of queries
    with workload matrix $\qmat^+$ so that the queries
    $q^+:\uni^+ \to \R$ of $\queries^+$
    are obtained by restricting queries $q:\uni \to \R$ of $\queries$
    to the input space $\uni^+$; define $\queries^-$ analogously.
    Then $\queries$, and also $\qmat$, are called symmetric.
\end{definition}

The following result will allow us
to translate our lower bound for the symmetric query workloads
into a lower bound for general query workloads.
Its proof is given in Appendix~\ref{ap:symmetrization}.

\begin{lemma}\label{lm:symtonotsym}
    Let $\alpha, \epsilon > 0$.
    Let $\queries$ be a symmetric workload of statistical queries
    and take $\queries^+$ as given by Definition~\ref{def:symmetric}.
    Suppose there exists
    a non-interactive locally $\eps$-LDP mechanism $\mech^+$
    which takes $n$ samples as input
    and achieves $\err^{\ell_\infty}(\mech^+, \queries^+, n) \le \alpha$.
    Then there exists
    a local $3\eps$-LDP mechanism $\mech$
    which takes $n' = \max\{n, \frac{1}{\eps^2 \alpha^2}\}$ samples as input
    and achieves $\err^{\ell_\infty}(\mech, \queries,n') \le 4 \alpha$.
\end{lemma}

Lemma~\ref{lm:symdual} allows us to relate
$\gamma_2(\qmat)$ and $\gamma_2(\qmat^+)$
and their witnesses.
Its proof is also given in Appendix~\ref{ap:symmetrization}.
\begin{lemma}\label{lm:symdual}
    Let $\alpha > 0$
    and let $\qmat \in \R^{\queries \times \uni}$
    be a symmetric workload matrix with $\uni^+$ and $\qmat^+$
    as given by Definition~\ref{def:symmetric}.
    Then it holds that $\gamma_2(\qmat) = \gamma_2(\qmat^+)$
    and $\gamma_2(\qmat,\alpha) = \gamma_2(\qmat^+,\alpha)$.
    Moreover, if, for some $U^+ \in \R^{\queries \times \uni^+}$,
    \[
        \gamma_2(\qmat^+,\alpha)
        = \frac{\qmat^+\bullet U^+ - \alpha\|U^+\|_1}{\gamma_2^\ast(U^+)},
    \]
    then
    \[
        \gamma_2(\qmat,\alpha)
        = \frac{\qmat\bullet U - \alpha\|U\|_1}{\gamma_2^\ast(U)},
    \]
    where $U = \frac12 (U^+,U^-)$ is a matrix in $\R^{\queries \times
      \uni}$ such that the submatrix $U^-$ is indexed by $\uni^-$ and
    has entries $u^-_{\query,-\elem} = -u^+_{\query, \elem}$
    for all $\elem \in \uni^+$ and $\query \in \queries$.
\end{lemma}

\subsection{Lower Bound based on Dual Solutions}

In this section we put together the different tools we have already
set up -- the KL-divergence lower bound, and the duality of the
approximate $\fnorminf$ norm -- in order to prove our main lower bound
result Theorem~\ref{thm:main-local}.

For this section, it is convenient to consider the enumeration
$\query_1,\dots,\query_{\qsize}$ of the queries of a symmetric
workload $\queries$ with workload matrix $\qmat \in \R^{[\qsize] \times
  \uni}$.  Let $U$ be the dual witness to the lower bound on
$\fnorminf(W,\alpha)$, as given by Lemma~\ref{lm:duality}, so that
\begin{equation}\label{eq:dualu}
    \gamma_2(\qmat, \alpha) 
    = \frac{\qmat\bullet U - \alpha\|U\|_1}{\gamma_2^\ast(U)}.
\end{equation}
By Lemma~\ref{lm:symdual}, we may assume
without loss of generality that $U$ is of the form $(U^+,U^-)$
where each entry of $U^-$ is the additive inverse
of the corresponding entry of $U^+$.
Furthermore, by dividing each entry of $U$ by $\|U\|_1$
if necessary,
then we may assume without loss of generality that $\|U\|_1=1$.
In this case,
\[
    \gamma_2(\qmat, \alpha) 
    = \frac{\qmat\bullet U - \alpha}{\gamma_2^\ast(U)}.
\]
Let us make a first attempt at constructing
our collection of ``hard'' distributions
$\dista_1,\dots,\dista_\qsize$ and $\distb_1,\dots,\distb_\qsize$ for $\queries$.
Since $\|U\|_1=1$, then
\begin{equation}\label{eq:pidef}
    \pi(v) = \sum_{\elem \in \uni}{|u_{v,\elem}|}
\end{equation}
defines a valid probability distribution over $[\qsize]$.
For each $v \in [\qsize]$, we then define a pair of distributions
$\dista_v$ and $\distb_v$ given by
\begin{align}
    \forall \elem \in \uni^+: \ \ &\dista_v(\elem) = \dista_v(-\elem) = |u_{v,\elem}|/\pi(v)
    \label{eq:pdef} \\
    \forall \elem \in \uni^+: \ \
        &\distb_v(\elem)
        = \begin{cases}
          2|u_{v,\elem}|/\pi(v) &\text{if }u_{v,\elem}\ge 0 \\
          0 &\text{if }u_{v,\elem}< 0
        \end{cases}
        \label{eq:qdef1} \\
        &\distb_v(-\elem)
        = \begin{cases}
          0 &\text{if }u_{v,\elem}\ge 0 \\
          2|u_{v,\elem}|/\pi(v) &\text{if }u_{v,\elem}< 0
        \label{eq:qdef2}
        \end{cases}
\end{align}

Then, for all $i,v \in [\qsize]$,
the symmetry of $\dista_v$ implies $\query_i(\dista_v) = 0$.
By contrast, it holds for all $v \in [\qsize]$ that
\begin{align*}
  \query_v(\distb_v) &= \sum_{\elem \in \uni} \query_v(\elem) \distb_v(\elem)\\
  &= \sum_{\elem \in \uni^+} \query_v(\elem) (\distb_v(\elem) - \distb_v(-\elem))\\
  &= 2\qmat^+ \bullet U^+ = \qmat \bullet U.
\end{align*}

Hence,
\[
    \ex{V \sim \pi}{\max_{i \in [\qsize]} \query_i(\distb_V)}
    \ge \ex{V \sim \pi}{\query_V(\distb_V)}
    = \qmat \bullet U.
\]
Since
$\qmat \bullet U = \gamma_2^*(U)\gamma_2(\qmat,\alpha) +\alpha \ge \alpha$
by Lemma~\ref{lm:duality},
then $\E_{V \sim \pi}[\max_{i \in [\qsize]} \query_i(\distb_V)] \ge \alpha$.
If we could guarantee that $\query_V(\distb_V)$ was close to its expectation
when $V \sim \pi$,
then estimating each of the queries $\query_i$ of $\queries$
with error less than $\alpha$
would allow us to distinguish the distributions
$\dista_1,\dots,\dista_\qsize$
from the distributions $\distb_1,\dots,\distb_\qsize$.
The following result modifies our distributions
in a way that resolves this issue.

\begin{lemma}\label{lm:hard-distros}
    Let $\queries$ be a collection of symmetric queries with
    workload matrix $W \in \R^{[\qsize] \times \uni}$.
    Let $U \in \R^{[\qsize] \times \uni}$
    be the dual witness so that \eqref{eq:dualu} is satisfied.
    Then there exist probability distributions
    $\widetilde{\dista}_1, \ldots, \widetilde{\dista}_\qsize$
    and $\widetilde{\distb}_1, \ldots, \widetilde{\distb}_\qsize$
    over $\uni$, and a distribution $\widetilde{\pi}$ over $[\qsize]$ 
    such that: 
    \begin{enumerate}
        \item $\query_i(\widetilde{\dista}_v) = 0$
            for all $i, v \in [\qsize]$; \label{crit:dista}
        \item for all $v$ in the support of $\widetilde{\pi}$,
            $
                \query_v(\widetilde{\distb}_v)
                \ge \frac{\qmat \bullet U - \alpha/4}{O(\log(1/\alpha))}
            $; \label{crit:distb}
        \item the matrix
            $
                \widetilde{U} \in \R^{[\queries] \times \uni}
            $
            with entries 
            $
                \widetilde{u}_{v,x}
                = \widetilde{\pi}(v) ( \widetilde{\dista}_v(x) - 
                    \widetilde{\distb}_v(x) )
            $
            satisfies
            $
                \fnorminf^\ast(\widetilde{U})
                \le \fnorminf^\ast(U).
            $
    \end{enumerate}
\end{lemma}

The proof of Lemma~\ref{lm:hard-distros}
will take advantage of the following exponential binning lemma.
\begin{lemma}\label{lm:binning}
    Suppose that $a_1, \ldots, a_\qsize \in [0,1]$
    and that $\pi$ is a probability distribution over $[\qsize]$.
    Then for any $\beta\in (0,1]$,
    there exists a set $S \subseteq [\qsize]$ such that
    $\pi(S) \cdot \min_{v\in S}a_v \ge \frac{\sum_{v=1}^\qsize\pi(v)a_v - \beta}{O(\log(1/\beta))}$.
\end{lemma}
\begin{proof}
Let $S_\ell = \{v: 2^{-\ell-1}< a_v \le 2^{-\ell}\}$
for $\ell \in \{0, \ldots, L\}$, where $L = \log_2(1/\beta)-1$,
and let $S_\infty = \{v: a_v \le \beta\}$.
Then, because $\sum_{v\in S_\infty}\pi(v)a_v \le \beta$, we have 
\[
\sum_{\ell = 1}^L\sum_{v \in S_\ell} \pi(v)a_v \ge 
\sum_{v=1}^{\qsize}\pi(v)a_v - \beta.
\]
Therefore, there exists $\ell$ such that
\[
    \sum_{v \in S_\ell} \pi(v)a_v
    \ge
    \frac{\sum_{v=1}^m\pi(v)a_v - \beta}{L}.
\]
The lemma now follows by taking $S = S_\ell$, since $\min_{v \in S_\ell}a_v \ge \frac12 \max_{v \in S_\ell}a_v$.
\end{proof}

\begin{proof}[Proof of Lemma~\ref{lm:hard-distros}]
    Let
    $\dista_1,\dots,\dista_\qsize$,
    $\distb_1,\dots,\distb_\qsize$,
    and $\pi$
    be as given by equations (\ref{eq:pidef}) - (\ref{eq:qdef2}).
    Since $q_v(\distb_v) > 0$ for all $v$,
    we may apply Lemma~\ref{lm:binning}
    with $a_v = q_v(\distb_v)$ and $\beta  = \alpha/4$
    to obtain a subset $S\subseteq [\qsize]$ for which
    \[
        \pi(S) \cdot \min_{v \in S} \query_v(\distb_v) 
        \ge \frac{\E_{V \sim \pi} \query_v(\distb_v)
            - {\alpha}/{4}}{O(\log(1/\alpha))}
        = \frac{\qmat \bullet U - {\alpha}/{4}}{O(\log(1/\alpha))}.
    \]
    Now define $\widetilde{\pi}$ as $\pi$ conditional on $S$.
    In particular,
    \[
        \widetilde{\pi}(v)
        =
        \begin{cases}
            \pi(v)/\pi(S), &\text{if }v \in S \\
            0, &\text{otherwise.}
        \end{cases}
    \]
    Then, for all $v \in [\qsize]$,
    define $\widetilde{\dista}_v = \dista_v$
    and $\widetilde{\distb}_v = \pi(S)\distb_v + (1-\pi(S))\dista_v$.
    This implies
    \[
        \forall i, v \in [\qsize]: \ \ 
        \query(\widetilde{\dista}_v) = \query({\dista}_v) = 0,
    \]
    \[
        \forall v \in [\qsize]: \ \
        \query_v(\widetilde{\distb}_v) = \pi(S) \query_v(\distb_v) 
        \ge 
        \frac{\qmat \bullet U - {\alpha}/{4}}{O(\log(1/\alpha))},
    \]
    \[
        \forall v \in [\qsize]: \ \
        \widetilde{\distb}_v - \widetilde{\dista}_v = \pi(S)(\distb_v - \dista_v).
    \]
    By the last of these facts,
    together with the definition of $\widetilde{\pi}$,
    it follows that the entries
    $
        \widetilde{u}_{v,x}
        =
        \widetilde{\pi}(v)
        ( \widetilde{\dista}_v(x) - \widetilde{\distb}_v(x) )
    $
    of the matrix $\widetilde{U}$ satisfy
    \[
        \widetilde{u}_{v,x}
        =
        \begin{cases}
            u_{v,x}, &\text{if }v \in S \\
            0, &\text{otherwise}.
        \end{cases}
    \]
    In other words, $\widetilde{U}$ is obtained from $U$
    by replacing some of its rows with the zero-vector.
    It is easy to see from 
    the definition of $\fnorminf^\ast$ that this implies
    $\fnorminf^\ast(\widetilde{U}) \le \fnorminf^\ast(U)$.
\end{proof}

Consider now the matrix $\widetilde{M} \in \R^{[\qsize] \times \uni}$
with entries
$
    \widetilde{m}_{v, \elem} = \widetilde{\dista}_v(\elem) - \widetilde{\distb}_v(\elem)
$.
Since $\widetilde{M}$ is obtained from the matrix $\widetilde{U}$
of Lemma~\ref{lm:hard-distros}
by scaling each row $v$ of $\widetilde{U}$ by $\frac{1}{2\pi(v)}$,
it follows that
\begin{equation*}
    \|\widetilde{M}\|_{\ell_\infty \to L_1(\widetilde{\pi})}
    = \frac{1}{2} \|\widetilde{U}\|_{\infty \to 1}
    \le \gamma^\ast_2(\widetilde{U})
    \le \gamma^\ast_2(U)
    = \frac{\qmat \bullet U - \alpha}{\gamma_2(\qmat, \alpha)}.
\end{equation*}
This is not quite the quantity 
\[
    \|\widetilde{M}\|_{\ell_\infty \to L_2(\widetilde{\pi})}^2
    =
    \max_{f \in \R^\uni: \|f\|_\infty \le 1} \ex{V \sim \pi}{\left(\E_{\elem \sim \widetilde{\dista}_V}[f_\elem] - \E_{\elem \sim \widetilde{\distb}_V}[f_\elem]\right)^2}
\]
which Lemma~\ref{lm:kl-div} would have us bound.
For comparison, note
\begin{equation*}
    \|\widetilde{M}\|_{\ell_\infty \to L_1(\widetilde{\pi})}
    = 
    \max_{f \in \R^\uni: \|f\|_\infty \le 1} \ex{V \sim \pi}{\left|\ex{x \sim \widetilde{\dista}_V}{f_x}] - \ex{\elem \sim \widetilde{\distb}_V}{f_\elem}\right|}.
\end{equation*}
Since the trivial case of Holder's inequality implies that
the $L_1(\widetilde{\pi})$-norm is always bounded above
by the $L_2(\widetilde{\pi})$-norm,
it holds that  
$
    \|\widetilde{M}\|_{\ell_\infty \to L_1(\widetilde{\pi})}
    \le
    \|\widetilde{M}\|_{\ell_\infty \to L_2(\widetilde{\pi})}
$.
However, this inequality goes in the wrong direction for our requirements.
This issue is remedied by taking advantage of
Lemma~\ref{lm:gammastar-inftyto2-nonuniform}.



\begin{lemma}\label{lm:hard-distros2}
    Let $\queries$ be a collection of symmetric queries with
    workload matrix $W \in \R^{[\qsize] \times \uni}$.
    Let $U \in \R^{[\qsize] \times \uni}$
    be the dual witness so that \eqref{eq:dualu} is satisfied.
    Then there exist probability distributions
    $\widetilde{\dista}_1, \ldots, \widetilde{\dista}_\qsize$
    and $\widetilde{\distb}_1, \ldots, \widetilde{\distb}_\qsize$
    over $\uni$, and a distribution $\widehat{\pi}$ over $[\qsize]$ 
    such that: 
    \begin{enumerate}
        \item
            $\widetilde{\dista}_1, \dots, \widetilde{\dista}_\qsize,
            \widetilde{\distb}_1, \dots, \widetilde{\distb}_\qsize$
            and $\widehat{\pi}$
            satisfy
            criteria~\ref{crit:dista}.~and~\ref{crit:distb}.~of Lemma~\ref{lm:hard-distros};
        \item
            the matrix $\widetilde{M}$ with entries 
            $
                \widetilde{m}_{v,x}
                = \widetilde{\dista}_v(x) - \widetilde{\distb}_v(x) 
            $
            satisfies
            \[
                \| \widetilde{M} \|_{\ell_\infty \to L_2(\widehat{\pi})}
                \le 4\fnorminf^\ast(U)
                = \frac{4(\qmat \bullet U - \alpha)}{\gamma_2(\qmat, \alpha)}
            \]
    \end{enumerate}
\end{lemma}
\begin{proof}
  Let $\widetilde{\dista}_1, \dots, \widetilde{\dista}_\qsize,
  \widetilde{\distb}_1, \dots, \widetilde{\distb}_\qsize$ and
  $\widetilde{\pi}$ be the distributions guaranteed to exist by
  Lemma~\ref{lm:hard-distros}, and let $\widetilde{U} \in \R^{[\qsize]
    \times \uni}$ be the corresponding matrix with entries $
  \widetilde{u}_{v,x} = \widetilde{\pi}(v) ( \widetilde{\dista}_v(x) -
  \widetilde{\distb}_v(x) ) $.  The entries of the matrix
  $\widetilde{M}$  satisfy $\pi(v)
  \widetilde{m}_{v,x} = \widetilde{u}_{v,x}$, so  we may apply
  Lemma~\ref{lm:gammastar-inftyto2-nonuniform} to obtain a
  distribution $\widehat{\pi}$ such that
    \[
        \| \widetilde{M} \|_{\ell_\infty \to L_2(\widehat{\pi})}
        \le 4\fnorminf^\ast(\widetilde{U})
        \le 4\fnorminf^\ast(U)
        = \frac{4(\qmat \bullet U - \alpha)}{\gamma_2(\qmat, \alpha)}.
    \]
    Lemma~\ref{lm:gammastar-inftyto2-nonuniform}
    further guarantees that the support of $\widehat{\pi}$
    lies within the support of $\tilde{\pi}$,
    which together with the properties of the distributions
    $\widetilde{\dista}_1, \dots, \widetilde{\dista}_\qsize,
    \widetilde{\distb}_1, \dots, \widetilde{\distb}_\qsize$
    and $\widetilde{\pi}$
    gives the first condition of our lemma.
\end{proof}

At last, we have all the components needed to prove our lower bounds
for symmetric workloads.

\begin{theorem}\label{lm:sqlbsym}
    Let $\alpha, \eps \in (0,1]$.
    Let $\queries$ be a symmetric workload of statistical queries
    with workload matrix $\qmat \in \mathbb{R}^{[\qsize] \times \uni}$.
    Then, for some
    $\alpha' = \Omega( \alpha / \log(1/\alpha) )$, 
    if
    $\frac{\gamma_2(\qmat, \alpha)^2}{\eps^2 \alpha^2} \ge 
    \frac{C\log 2\qsize}{(\alpha')^2}$ for a large enough constant
    $C$, we have
    \[
    \sc_{\priv,0}^{\ell_\infty, \mathrm{loc}}(\queries, \alpha')
    = \Omega \left( \frac{\gamma_2(\qmat, \alpha)^2}{\eps^2 \alpha^2} \right).
    \]    
\end{theorem}
\begin{proof}
  Let $\alpha' = \Omega( \alpha / \log(1/\alpha) )$ be a value that
  will be decided shortly, and $C'$ be a sufficiently large constant. 
  If we run a $\priv$-DP mechanism $\mech$ on
  $\dsize =  \max\left\{\sc^{\ell_\infty}(\mech, \queries, \alpha'), \frac{C'\log
      2\qsize}{(\alpha')^2}\right\}$ samples drawn i.i.d.~from some
  distribution $\dist$ on $\uni$, then, by 
  classical uniform convergence results, 
  \(
  \ex{\ds \sim\dist^\dsize}{\left\|\queries(\ds)-\queries(\dist)\right\|_\infty}
  \le {\alpha'},
  \)
  where $\queries(\dist) = (q_1(\dist), \ldots,
  q_\qsize(\dist))$. Therefore, the mechanism will satisfy 
  \begin{equation}
    \label{eq:pop-err}
    \ex{\ds \sim
      \dist^\dsize}{\left\|\mech(\ds)-\queries(\dist)\right\|_\infty} 
    \le 2\alpha'.
  \end{equation}
  We will show that for any $\eps$-LDP mechanism $\mech$ such that
  \eqref{eq:pop-err} holds for an arbitrary $\dist$, we must have 
  \begin{equation}
    \label{eq:lb-pop}
    \dsize =
    \Omega\left( \frac{\gamma_2(\qmat, \alpha)^2}{\eps^2 \alpha^2}
    \right).    
  \end{equation}
  Therefore, we get that 
  \(
  \max\left\{ \sc^{\ell_\infty}(\mech, \queries, \alpha'), 
    \frac{C'\log  2\qsize}{(\alpha')^2}\right\}
  = 
  \Omega\left( \frac{\gamma_2(\qmat, \alpha)^2}{\eps^2 \alpha^2},
  \right)\)
  which implies the theorem by the assumption on
  $\frac{\gamma_2(\qmat, \alpha)^2}{\eps^2 \alpha^2}$.

    Let
    $\widetilde{\dista}_1,\dots,\widetilde{\dista}_\qsize$,
    $\widetilde{\distb}_1,\dots,\widetilde{\distb}_\qsize$
    and $\widehat{\pi}$ be the distributions,
    and $\widetilde{M} \in \R^{[\qsize] \times \uni}$ the matrix,
    guaranteed to exist by Lemma \ref{lm:hard-distros2}.
    The matrix $\widetilde{M}$ has entries
    $
        \widetilde{m}_{v,x}
        = {\widetilde{\dista}}_v(x) - {\widetilde{\distb}}_v(x)
    $
    and satisfies
    \[
        \| \widetilde{M} \|_{\ell_\infty \to L_2(\widehat{\pi})}
        \le \frac{4(\qmat \bullet U - \alpha)}{\gamma_2(\qmat, \alpha)}.
    \]
    Equivalently,
    \[
        \max_{f \in \R^\uni: \|f\|_\infty \le 1} \ex{V \sim \widehat{\pi}}{\left(\ex{x \sim \widetilde{\dista}_V}{f_x} - \ex{x \sim \widetilde{\distb}_V}{f_x}\right)^2}
        \le 
        \left( \frac{4(\qmat \bullet U - \alpha)}{\gamma_2(\qmat, \alpha)} \right)^2.
    \]
    By Lemma \ref{lm:kl-div}, this implies
    \begin{equation}\label{eq:divbound}
        \div(\trans(\widetilde{\dista}_{\widehat{\pi}}^n)\|\trans(\widetilde{\distb}_{\widehat{\pi}}^n)) \le O(n \eps^2)\cdot 
        \left( \frac{\qmat \bullet U - \alpha}{\gamma_2(\qmat, \alpha)} \right)^2
    \end{equation}

    Lemma \ref{lm:hard-distros2} guarantees further that
    $q_i(\widetilde{\dista}_v) = 0$ for all $i,v \in [\qsize]$, while
    $ q_v(\widetilde{\distb}_v) \ge \frac{\qmat \bullet U - \alpha /
      4}{O(\log ( 1 / \alpha ) )} $ for all $v$ in the support of
    $\widehat{\pi}$.  Let 
    \(
    \alpha' = \frac{1}{8} \min_{v \in   [\qsize]} q_v(\widetilde{\distb}_v).  
    \)
    Then a mechanism $\mech$ satisfying \eqref{eq:pop-err} can distinguish
    between the distributions $\widetilde{\dista}_{\widehat{\pi}}^n$
    and $\widetilde{\distb}_{\widehat{\pi}}^n$ with constant
    probability, and, by Pinsker's inequality,   
    $
    \div(\trans(\widetilde{\dista}_{\widehat{\pi}}^n)\|\trans(\widetilde{\distb}_{\widehat{\pi}}^n))
    $
    is bounded from below by some constant $C > 0$.
    By (\ref{eq:divbound}),
    this implies that
    \[
        n =
        \Omega\left(
            \frac{\gamma_2(\qmat, \alpha)}{\eps \cdot (\qmat \bullet U - \alpha)} 
        \right)^2
    \]
    samples are required to obtain
    accuracy $\alpha' / 4$ and privacy $\priv$.

    \ul{Case 1: $\qmat \bullet U \le 2 \alpha$}. 
    Recall that $\qmat \bullet U \ge \alpha$.
    Hence, if $\qmat \bullet U \le 2 \alpha$,
    then
    $
        n = \Omega \left(
            \frac{\gamma_2(\qmat, \alpha)^2}{\eps^2 \alpha^2} 
        \right)
    $
    and furthermore
    \[
        \alpha' 
        \ge
        \frac{\qmat \bullet U - \alpha/4}{O(\log(1/\alpha))}
        =
        \Omega \left( \frac{\alpha}{\log ( 1 / \alpha )} \right)
    \]

    \ul{Case 2: $\qmat \bullet U > 2 \alpha$.}
    However, if $\qmat \bullet U > 2 \alpha$,
    then, for $\beta \in [0,1]$,
    we may instead consider the distributions
    $
        \widehat{\distb}_v = (1 - \beta) \cdot \widetilde{\dista}_v + \beta \cdot \widetilde{\distb}_v
    $
    and $\widehat{\dista}_v = \widetilde{\dista}_v$, given for $v \in [\qsize]$.
    We have
    \begin{align*}
      \div(\trans(\widehat{\dista}_{\widehat{\pi}}^n)\|\trans(\widehat{\distb}_{\widehat{\pi}}^n))
        & \le O(\eps^2 n) \cdot \max_{f \in \R^\uni: \|f\|_\infty \le  1} \ex{V \sim \widehat{\pi}}{\left(\ex{x \sim \widehat{\dista}_V}{f_x} - \ex{x \sim \widehat{\distb}_V}{f_x}\right)^2} \\
        & = O(\eps^2 n) \cdot \beta^2 \max_{f \in \R^\uni: \|f\|_\infty \le  1} \ex{V \sim \widehat{\pi}}{\left(\ex{x \sim \widetilde{\dista}_V}{f_x} - \ex{x \sim \widetilde{\distb}_V}{f_x}\right)^2} \\
        &\le O(\eps^2 n) \cdot \beta^2 \cdot \left( \frac{\qmat \bullet U - \alpha}{\gamma_2(\qmat, \alpha)} \right)^2.
    \end{align*}

    Also, $q_i(\widehat{\dista}_v) = 0$ for all $i, v \in [\qsize]$,
    while
    \[
        q_v(\widehat{\distb}_v) = \beta\cdot q_v(\widetilde{\distb}_v)
        \ge
        \beta\cdot \left(
            \frac{\qmat\bullet U-\alpha/4}{O(\log(1/\alpha))}
        \right)
    \]
    for all $i$ in the support of $\widehat{\pi}$.
    In particular,
    if we set
    $\alpha' = \frac{1}{8} \min_v q_v(\widehat{\distb}_v) \ge
    \frac{\beta(\qmat \bullet U - \alpha / 4)}{O(\log ( 1 / \alpha )
      )}$,
    and \eqref{eq:pop-err} holds for $\mech$ and this value of
    $\alpha'$,
    then $\mech$ can distinguish between $\widehat{\dista}_{\widehat{\pi}}^n$ and $\widehat{\distb}_{\widehat{\pi}}^n$.
    This implies
    $\div(\trans(\widehat{\dista}_{\widehat{\pi}}^n)\|\trans(\widehat{\distb}_{\widehat{\pi}}^n))$
    is bounded below by a constant,
    from which we obtain that
    \[
        n
        = \Omega \left( \frac{\gamma_2(\qmat, \alpha)}{\eps \beta \cdot (\qmat\bullet U)} \right)^2
    \]
    samples are required for privacy $\eps$ and accuracy $\alpha'$.
    Indeed,
    by taking $\beta = \frac{U \bullet \qmat}{\alpha}$,
    we get that
    \[
        n = \Omega \left( \left( \frac{\gamma_2(\qmat, \alpha)}{\eps \alpha} \right)^2 \right)
    \]
    samples are required
    for privacy $\eps$ and accuracy $\alpha'$
    which satisfies
    \(
        \alpha'
        \ge
        \frac{\beta(\qmat \bullet U - \alpha / 4)}{O ( \log ( 1 / \alpha ) )}
        =
        \Omega \left( \frac{\alpha}{\log ( 1 / \alpha )} \right).
    \)

    In both cases,
    $
        n = \Omega \left(
            \frac{\gamma_2(\qmat, \alpha)^2}{\eps^2 \alpha^2}
        \right)
    $
    samples are required
    for privacy $\eps$ and accuracy
    $\alpha'$,
    where
    $
        \alpha' = \Omega \left(\frac{\alpha}{\log ( 1 / \alpha )} \right)
    $    
\end{proof}

\begin{theorem}[Formal version of Theorem~\ref{thm:main-local}]\label{thm:sqlbfinal}

   Let $\alpha, \priv \in (0,1]$.
   Let $\queries$ be a collection of queries
   with workload matrix $\qmat$.
   Then, for some
   $\alpha' = \Omega \left( \frac{\alpha}{\log(1 / \alpha)} \right)$, if
    $\frac{\gamma_2(\qmat, \alpha)^2}{\eps^2 \alpha^2} \ge 
    \frac{C\log 2\qsize}{(\alpha')^2} + \frac{C}{\priv^2(\alpha')^2 }$ for a large enough constant
    $C$, we have
    \[
    \sc_{\priv,0}^{\ell_\infty, \mathrm{loc}}(\queries, \alpha')
    = \Omega \left( \frac{\gamma_2(\qmat, \alpha)^2}{\eps^2 \alpha^2} \right).
    \]        
\end{theorem}

\begin{proof}
    Let $\queries$ be the workload of queries
    with workload matrix $\qmat$
    and let $\queries'$ be the symmetric workload of queries
    with workload matrix of the form $(\qmat,-\qmat)$. Recall that, by
    Lemma~\ref{lm:symdual}, $\fnorminf(\qmat', \alpha) =
    \fnorminf(\qmat, \alpha)$. 

    For arbitrary $\alpha' > 0$,
    consider a mechanism $\mech$ for $n$ agents
    which achieves
    \[
    \err^{\ell_\infty}(\mech, \queries,n) \le \alpha'.
    \]
    In this case,
    Lemma \ref{lm:symtonotsym} guarantees the existence of
    a mechanism $\mech'$
    which takes $2n + \frac{1}{\eps^2 (\alpha')^2}$ samples as input
    and achieves $\err^{\ell_\infty}(\mech', \queries',n) \le 4\alpha'$.

    Moreover,
    for $\alpha > 0$, Theorem~\ref{lm:sqlbsym} says that
    we can choose $\alpha' = \Omega \left( \frac{\alpha}{\log(1 / \alpha)} \right)$
    so that our sample complexity guarantee for $\mech'$ implies
    \[
        2n + \frac{1}{\eps^2 (\alpha')^2}
        = \Omega \left( \frac{ \gamma_2(\qmat', \alpha)^2 }{\eps^2 \alpha^2} \right).
    \]
    Thus, for some constant $C > 0$,
    $\gamma_2(\qmat',\alpha) =\gamma_2(\qmat,\alpha) \ge \frac{C}{(\alpha')^2 \priv^2}$ 
    implies
    $
        n = \Omega \left(
            \frac{\gamma_2(\qmat',\alpha)}{\alpha^2 \eps^2}
        \right)
        = \Omega\left(\frac{\gamma_2(\qmat,\alpha)}{\alpha^2 \eps^2}\right)
    $.
\end{proof}

\subsection{Applications of the Lower Bounds}

In this subsection we apply Theorem~\ref{thm:sqlbfinal} to several
workloads of interest, and, using known bounds on the approximate
$\fnorminf$ norm, prove new lower bounds on the sample complexity of these
workloads. 

We start with the threshold queries
$\queries^{\textrm{cdf}}_{T}$. Identifying $q_t$ with $t$, we see
that the corresponding workload matrix $\qmat$ is a lower triangular
matrix, with entries equal to $1$ on and below the main diagonal. Let
us consider a different matrix $\qmat' = 2\qmat - J$, where $J$ is the
all-ones $\usize \times \usize$ matrix. Forster et al.~\cite{ForsterSSS03} showed
a lower bound on the margin complexity of $\qmat'$, which implies that
for any $\widehat{\qmat}$ such that $\widehat{w}_{t,\elem}
w'_{t,\elem} \ge 1$ holds for all $t,\elem \in [T]$, we have
\begin{equation}\label{eq:thresh-lb}
\fnorminf(\widehat{\qmat}) = \Omega(\log \usize). 
\end{equation}
Note that, if $\widetilde{\qmat}$ satisfies $\|\widetilde{\qmat} -
\qmat'\|_{1\to\infty} \le \frac12$, then we can take $\widehat{\qmat}
= 2\widetilde{\qmat}$, and \eqref{eq:thresh-lb} implies
\(
\fnorminf(\qmat', 1/2) = \Omega(\log \usize).
\)
Finally, homogeneity and the triangle inequality for $\fnorminf$, and
$\fnorminf(J) = 1$ imply that 
\(
\fnorminf(\qmat, 1/2) 
\ge \frac12\fnorminf(\qmat', 1/2) - \frac12
= \Omega(\log \usize).
\)
Together with Theorem~\ref{thm:sqlbfinal}, this gives
Corollary~\ref{cor:threshold}. 

Next, we consider the parity queries
$\queries^{\textrm{parity}}_{d,w}$. Note that the workload matrix
$\qmat$ of these queries is a submatrix consisting of ${d \choose w}$
rows of the $2^d \times 2^d$ Hadamard matrix. Let $s = 2^d {d \choose  w}$ be the number of entries in $\qmat$. To prove a lower bound on
$\fnorminf(\qmat, \alpha)$, we can use Lemma~\ref{lm:duality} with $U = \qmat$. The
rows of a Hadamard matrix are pairwise orthogonal and have $\ell_2$
norm $2^{d/2}$, and, so, Lemma~\ref{lm:gammastar-dual}, used with $P$ and $Q$
set to appropriately scaled copies of the identity matrices of the
respective dimensions, implies that $\fnorminf^*(U) \le
\sqrt{s2^{d}}$. Moreover, $W\bullet U = \|U\|_1 = s$, and, by
Lemma~\ref{lm:duality}, we have
\[
\fnorminf(\qmat, 1/2) \ge \frac{\sqrt{s}}{2^{(d/2) + 1}}
= \Omega\left({d \choose  w}^{1/2}\right).
\]
This gives Corollary~\ref{cor:parity}. 

Finally, we treat marginal queries. Let us define these queries slightly
more generally than we did in the introduction, by allowing for
negation. We define $\queries^{\textrm{marginal}}_{d,w}$ to consist of
the queries $q_{S, y}(\ds) = \frac{1}{n} \sum_{i=1}^{n} \prod_{j \in
  S} \mathbb{I}[\dsrow_{i,j} = y_j]$, with $S$ ranging over subsets of $[d]$
of size at most $w$, and $y$ ranging over $\{0,1\}^d$. These queries
can be expressed in terms of the $q_S$ queries defined in the
introduction by doubling the dimension $d$. 

To prove a lower bound for $\queries^{\textrm{marginal}}_{d,w}$, we
use the pattern matrix method of Sherstov~\cite{Sherstov11}. We will omit a full
definition of a pattern matrix here, and refer the reader to
Sherstov's paper. Instead, we remark that, denoting by $f$ the
AND function on $w$ bits, a $(d, w, f)$-pattern matrix $\qmat'$ is a
$\frac{(2d)^w}{w^w} \times 2^d$ submatrix of the workload matrix
$\qmat$ for $\queries^{\textrm{marginal}}_{d,w}$. Let $s =2^d
\frac{(2d)^w}{w^w}$ be the number of entries in $\qmat'$.  By
Theorem~8.1.~in~\cite{Sherstov11}, we have that, for any $\alpha \le \frac16$,
\[
\min\left\{\frac{1}{\sqrt{s}}
\|\widetilde{\qmat} \|_{tr}: \|\widetilde{\qmat} -\qmat'\|_{1\to \infty} \le \alpha
\right\}
=\Omega\left(\frac{d}{w}\right)^{\mathrm{deg}_{1/3}(f)/2},
\]
where $\|\widetilde{\qmat}\|_{tr}$ is the trace-norm, i.e., the sum of
singular values of $\widetilde{\qmat}$, and $\mathrm{deg}_{1/3}(f)$ is
the $(1/3)$-approximate degree of $f$, which is known to be $
\Omega(\sqrt{w})$~\cite{NisanS94}. Since
$\frac{1}{\sqrt{s}}\|\widetilde{\qmat} \|_{tr}$ is a lower bound on
$\fnorminf(\widetilde{\qmat})$ (see~\cite[Lemma 3.4]{LinialMSS07}), this implies
\[
\fnorminf(\qmat, 1/6) \ge 
\fnorminf(\qmat', 1/6) =
\Omega\left(\frac{d}{w}\right)^{\Omega(\sqrt{w})},
\]
giving us Corollary~\ref{cor:marginal}.


\section{Non-Interactive Local DP: PAC Learning}

It turns out that we are able to translate
our algorithm and lower bound for answering linear queries in the local model
into an algorithm and lower bound for
\emph{probably approximately correct learning} in the local model.

A concept
$\concept:\uni^+ \rightarrow \{-1,+1\}$
from a concept class $\concepts$
identifies each sample $x$ of $\uni^+$ with a label $c(x)$.
The labelled pair $(x,c(x)) = (x,1)$
may be identified with the sample $x$ of $\uni^+$,
while the labelled pair $(x,c(x)) = (x,-1)$
may be identified with the sample $-x$ of $\uni^-$.
Let $q:\uni \to \{-1,+1\}$ be given by
\[
    q(x) =
    \begin{cases}
        \concept(\elem), &\text{if } \elem \in \uni^+ \\
        -\concept(-\elem), &\text{if } \elem \in \uni^-
    \end{cases}
\]
Then the \emph{loss} of the concept $\concept$ on a dataset
$\overline{X} = ((x_1,y_1),\dots,(x_\dsize,y_\dsize))$,
denoted $\loss_{\overline{X}}(c)$,
is
\begin{align*}
    \loss_{\overline{X}}(c)
    &= \frac{1}{\dsize} \sum_{i = 1}^{\dsize} (1 - \I[f(x_i) = y_i]) \\
    &= \frac{1}{2} - \frac{1}{2\dsize} \sum_{i = 1}^{\dsize} f(x_i) y_i
    = \frac{1}{2} - \frac{1}{2\dsize} \sum_{i = 1}^{\dsize} q(x_i \cdot y_i)
    = \frac{1}{2} - \frac{1}{2} \query(X)
\end{align*}
where $X$ is the dataset
$(x_1 \cdot y_1,\dots,x_\dsize \cdot y_\dsize)$.
In this way, estimating $\loss_{\overline{X}}(c)$
given the dataset $\overline{X}$
is equivalent to estimating $q(X)$
given the dataset $X$.
More generally,
if we consider the query workload $\queries$
consisting of all such queries $\query$
obtained from some concept $c$ of $\concepts$
in this way,
then estimating $\queries(X)$
is equivalent to estimating
$
    \loss_{\overline{X}}(\concepts)
    = (\loss_{\overline{X}}(\concept))_{\concept \in \concepts}
$.
This idea allows us to adapt the algorithm of
Theorem~\ref{thm:local-apxfact}
for estimating linear queries to an algorithm for learning.
The result is stated in terms of the \emph{concept matrix}
$\cmat \in \R^{\concepts \times \uni^+}$ of $\concepts$ with entries given by
\[
    d_{c,x} = c(x)
\]
and takes advantage of the fact that the workload matrix $W$ of the
corresponding query workload $\queries$ is obtained by extending $D$
to $\concepts \times \uni$ in the usual way with $w_{q,x} = d_{c,x}$
and $w_{q,-x} = -d_{c,x}$ for $q \in \queries$ and $x \in \uni^+$ when
$c$ is the concept that corresponds to $q$.  In particular, the
queries $\queries$ are symmetric, and, by Lemma~\ref{lm:symdual},
$\fnorminf(D,\alpha) = \fnorminf(W,\alpha)$.

In order to state our results for agnostic learning,
we need to define notation for population loss, in addition to the
empirical loss defined above.
For a distribution $\dist$ over $\uni^+ \times \{-1,+1\}$,
we will use $\loss_{\dist}(c)$
to denote the loss of the concept $c$ on $\dist$,
given by
\[
    \loss_{\dist}(c)
    = \Pr_{(x,y) \sim \dist} \left[ c(x) \neq y \right].
\]
For $\alpha,\beta > 0$ we will say that the mechanism $\mech$
\emph{($\alpha$,$\beta$)-learns} $\concepts$
with $n$ samples if,
for all distributions $\dist$ over $\uni^+ \times \{-1,+1\}$, 
given as input a dataset $\overline{X} =
((x_1,y_1),\dots,(x_\dsize,y_\dsize))$
of $n$ samples drawn IID from $\dist$,
$\mech$ outputs a concept $c \in \concepts$ and an estimate
$\overline{\loss}$ such that
\[
    \Pr_{\mech,\overline{X}}[
        \loss_{\dist}(c) \le \min_{\concept' \in \concepts} \loss_{\dist}(c') + \alpha
        \text{ and }
        |\overline{\loss} - \loss_{\dist}(c)| \le \alpha
    ]
    \ge 1 - \beta.
\]

Typically, the learning problem does not require outputting an
estimate of the loss $\loss_{\dist}(c)$, since it is usually easy to
compute such an estimate with few additional samples, once a concept
$c$ has been computed. In the local model, however, this would require
an additional round of interactivity. Since we focus on the
non-interactive local model, it is natural to make this additional
requirement on the learning algorithm. 

Since we wish to bound population loss,
it is necessary to assume that there are sufficiently many samples
to guarantee uniform convergence.
It suffices to assume, for some constant $C$,
that the number of samples is at least
$n \ge \frac{C\log 2|\concepts|}{\alpha^2}$ 
to guarantee 
\[
    \Pr_{\overline{X}}
    [\forall \concept \in \concepts, \ |\loss_{\overline{X}}(c) - \loss_{\dist}(c)| \le \alpha]
    \ge 1 - \frac{\beta}{2} 
\]
when $\overline{X}$ consists of $n$ IID samples
drawn from $\mu$.

\begin{theorem}

    Let $\alpha, \beta \in (0,1)$, and let $\priv > 0$.
    There exists
    an $\priv$-LDP mechanism $\mech$
    such that,
    for any concept class $\concepts$ of size $|\concepts| = \qsize$
    with corresponding concept matrix
    $D \in \R^{\concepts \times \uni^+}$,
    it suffices to have
    a dataset 
    $
        \overline{\ds}
        = ((x_1,y_1),\dots,(x_\dsize,y_\dsize))
    $
    of
    \[
        \dsize
        = \max \left\{
            O\left( \frac{\gamma_2(D,\alpha)^2 \log \qsize}{\priv^2 \alpha^2} \right),
            O\left( \frac{\log \qsize}{\alpha^2} \right)
        \right\}
    \]
    samples to guarantee that
    $\mech$ $(\alpha,\beta)$-learns $\concepts$.
    
\end{theorem}

Applying the same ideas,
we know that if we estimate
$\min_{\concept \in \concepts} \loss_{\overline{\ds}}(\concept)$,
then we can estimate
$\max_{\query \in \queries} q(X)$.
Similarly, estimating
$\min_{\concept \in \concepts} \loss_{\dist}(\concept)$
is equivalent to estimating
$\max_{\query \in \queries} q(\dist')$
where $\dist'$ is the distribution on $\uni$ obtained from $\dist$
by associating samples of the forms $(x,1)$ and $(x,-1)$
with $x$ and $-x$, respectively.
Since the matrix $\qmat \in \R^{\concepts \times \uni}$ obtained from
$D$ is symmetric, and estimating $\max_{\query \in \queries} q(\dist')$
is precisely what is required for the lower bound of Theorem~\ref{lm:sqlbsym},
we the following lower bound for agnostic learning.
\begin{theorem}

    Let $\beta \in (0,1)$ be a small enough constant, and let $\priv > 0$.
    Let $\concepts$ be a concept class
    with concept matrix $D \in \R^{\concepts \times \uni^+}$.
    For some
    $
        \alpha' = \Omega \left( \frac{\alpha}{\log(1/\alpha)}  \right)
    $,
    if
    $
        \frac{\gamma_2(\qmat,\alpha)}{\priv^2 \alpha^2}
        \ge \frac{C\log{2k}}{\alpha'}
    $
    for a large enough constant $C > 0$,
    then any $\priv$-LDP mechanism $\mech$
    which $(\alpha',\beta)$-learns $\concepts$
    requires
    \[
        n = \Omega \left( \frac{\gamma_2(W,\alpha)^2}{\priv^2 \alpha^2} \right)
    \]
    samples as input.

\end{theorem}


\section{Characterizing Central DP for Large Datasets}
The goal of this section is to show that the sample complexity of releasing a given set of linear queries with workload matrix $\qmat$ is
$$
\scz(\qmat, \alpha, \priv, \privd) = \Theta\left(\frac{\fnorm(\qmat)}{\alpha \priv}\right)
$$
when $\alpha$ is sufficiently small (smaller than some $\alpha^*( \queries, \priv )$).  Or, equivalently, we show that $\errz(\qmat,\dsize, \priv, \privd) = \Theta(\frac{\fnorm(\qmat)}{\priv \dsize})$, when $\dsize$ is sufficiently large (larger than some $n^*(\queries, \priv)$).

The proof consists of two steps.  First, we argue that error $\err(\qmat, \dsize, \priv, \privd) = \Theta(\frac{\fnorm(\qmat)}{\priv \dsize})$ is necessary for \emph{every} $n$ if we restrict attention only to mechanisms that are \emph{data-independent}.  That is, mechanisms that perturb the output with noise from a fixed distribution independent of the dataset.  Then, we apply a lemma of Bhaskara et al.~\cite{BhaskaraDKT12} that says, when $n$ is sufficiently large, any instance-dependent mechanism can be replaced with an instance-independent mechanism with the same error and similar privacy parameters.

\subsection{Data-Independent Mechanisms}

Let $\queries$ be a workload of linear queries over data universe $\uni$ and let $\qmat \in \R^{\queries \times \uni}$ be the matrix form of this workload.  An \emph{instance-independent} mechanism $\mech$ can be written (as a function of the histogram of the dataset) as, 
$$
\mech(\hist) = \frac{1}{n} (\qmat \hist + Z)
$$
where $Z$ is a random variable over $\R^{\queries}$ whose distribution does not depend on $\hist$.  Without loss of generality, we assume $\ex{}{Z}=0$.  Let $\Sigma = \ex{}{ZZ^T}$ be the covariance matrix of $Z$.  Then the $\ell_2$ error of such a mechanism is
\begin{align*}
\errz(\mech, \qmat,\dsize) ={} &\max_{\hist : \|\hist\|_1 = \dsize} \sqrt{ \ex{}{ \frac{\| \mech(\hist) - \frac{1}{n} \qmat \hist \|_2^2}{|\queries|}  } } 
={}  \sqrt{\ex{}{ \frac{ \| Z \|_2^2}{ \dsize |\queries|} }} ={} \sqrt{\frac{\tr(\Sigma)}{ \dsize |\queries|}}
\end{align*}
In this section, we will show that, if $\mech$ is $(\priv,\privd)$-differentially private (for $\priv,\privd$ smaller than some absolute constants), then $\tr(\Sigma) = \Omega(\frac{|\queries| \fnorm(\qmat)^2}{\priv^2})$, and thus $\err(\mech,\qmat, n) = \Omega(\frac{\fnorm(\qmat)}{\priv \dsize })$.

We start with the following basic lemma about differential privacy, which says that the variance of any differentially private algorithm for answering a single query $w$ must be proportional to the sensitivity of the query.
\begin{lemma}[\cite{KasiviswanathanRSU10}] \label{lem:single-query-variance}
For any single-query workload $w \in \R^{|\uni|}$, and any data-independent mechanism $\mech(\hist) = \frac{1}{n} w^\top \hist + \frac{1}{n} z$ that is $(\priv,\privd)$-differentially private for $\priv,\privd$ smaller than some absolute constants, $\ex{}{z^2} \geq \frac{1}{C\priv} \| w \|_\infty$
for some absolute constant $C > 0$.
\end{lemma}

Next, we define the \emph{sensitivity polytope} $K = \qmat B_1^{|\uni|}$, where $B_1^{|\uni|} = \{ h \in \R^{|\uni|} : \|h\|_1 \leq 1\}$.  With this definition, we have that for any pair of neighboring datasets $\ds, \ds'$ with associated histograms $\hist, \hist'$, we have $\qmat (\hist - \hist') \in K$.  The next lemma says that the covariance matrix $\Sigma$ defines an ellipsoid that contains at least a constant multiple of the sensitivity polytope.

\begin{lemma} \label{lem:workload-variance}
Let $W$ be a workload matrix such that the sensitivity polytope $K$ is full dimensional.  Let $\mech$ be an $(\priv,\privd)$-differentially private data-independent mechanism for $\qmat$ that has covariance matrix $\Sigma$, for $\priv,\privd$ smaller than some absolute constants.   Then $\Sigma$ is invertible, and
$$
\max_{y \in K} \| \Sigma^{-1/2} y\|_2^2 = \max_{y \in K} y^{\top} \Sigma^{-1} y \leq C^2 \priv^2
$$
for some absolute constant $C > 0$.
\end{lemma}
\begin{proof}
By post-processing, for any $u \in \R^{|\uni|}$, $$ u^\top \mech(\hist) = \frac{1}{n} u^\top \qmat \hist + \frac{1}{n} u^\top Z$$ is an $(\eps, \delta)$-DP mechanism for the single query $u^\top \qmat$. The sensitivity polytope of the workload $u^\top W$ is the line $[-h_K(u), h_K(u)]$, where $h_K(u) = \max_{y \in K}{u^\top y}$ is the support function.  By Lemma~\ref{lem:single-query-variance}, if $\mech$ is an $(\priv,\privd)$-differentially private mechanism, then for some constant $C$, 
\begin{equation}\label{eq:dir-var}
\| \Sigma^{1/2} u \|_2 = \sqrt{u^\top \Sigma u} \geq \frac{h_K(u)}{C\priv}.
\end{equation}
If $K$ is full dimensional, then in particular we have $h_{K}(e_i) > 0$ for any standard basis vector $e_i$, which implies that the matrix $\Sigma$ is positive definite and invertible.

By change of variables we can write $v = \Sigma^{1/2} u$ and rewrite \eqref{eq:dir-var} as
$$
\| v \|_2 \geq \frac{1}{C \priv} \cdot h_{K}(\Sigma^{-1/2} v) = \frac{1}{C \priv} \cdot \max_{y \in K} (\Sigma^{-1/2} v)^\top y = \frac{1}{C \priv} \cdot \max_{y \in K} v^\top \Sigma^{-1/2} y
$$
Since the above holds for any unit vector $v \in \mathbb{S}^{|\uni|-1}$, we have
$$
\max_{y \in K} \| \Sigma^{-1/2} y \|_2 =  \max_{y \in K}\max_{v \in \mathbb{S}^{|\uni|-1}}{v^\top \Sigma^{-1/2} y} = \max_{v \in \mathbb{S}^{|\uni|-1}}\max_{y \in K}{v^\top \Sigma^{-1/2} y} \leq C \priv
$$
where the first equality is the equality-case of Cauchy-Schwarz.  
\end{proof}

Recall that for a matrix $\qmat \in \R^{\queries \times \uni}$,
$$\fnorm(\qmat) = \inf \left\{ \tfrac{1}{ |\queries|^{1/2} }\|R\|_F \|A\|_{1 \to 2}: RA = \qmat \right\}.$$ 
We now prove our main result, which shows that the error of data-independent private mechanisms must be proportional to $\fnorm(\qmat)$.

\begin{theorem} \label{thm:indep-lb}
Let $\qmat$ be a workload matrix.  Let $\mech$ is a $(\priv,\privd)$-differentially private data-independent mechanism for $\qmat$ with covariance matrix $\Sigma$, for $\priv,\privd$ smaller than some absolute constants.  Then
$$\err^{\ell_2^2}(\mech,\qmat, n) = \Omega\left(\frac{\fnorm(\qmat)}{C \priv n}\right).$$
\end{theorem}
\begin{proof}
Let $w_1,\dots,w_{|\uni|}$ be the columns of the workload matrix $\qmat$.  Let $A = \Sigma^{-1/2} \qmat$ with columns $a_1,\dots,a_{|\uni|}$ and let $R = \Sigma^{1/2}$ so that $RA = \qmat$.  By Lemma~\ref{lem:workload-variance}, the matrix $A$ is well defined, and for every $i$, $\| a_i \| = \| \Sigma^{-1/2} w_i \|_2 \leq C \priv$.  Hence $\|A\|_{1 \to 2} \leq C \priv$.  We also have
\[
\|R\|_F = \tr(R^\top R)^{1/2} = \tr(\Sigma)^{1/2} = |\queries|^{1/2}
\cdot n\cdot \err^{\ell_2^2}(\mech, W, \dsize).
\]
Combining the inequalities, we get 
\[
\fnorm(W) \le \frac{1}{|\queries|^{1/2}} \|R\|_F \|A\|_{1\to 2} \le C
\priv \cdot n \cdot \err^{\ell_2^2}(\mech, W, \dsize).
\]
The theorem follows from rearranging this inequality.
\end{proof}

\subsection{From Data-Dependent to Data-Independent Mechanisms}

In this section we describe a reduction of Bhaskara et al.~\cite{BhaskaraDKT12} showing that any data-dependent mechanism with small error for datasets of arbitrary size can be converted into a data-independent mechanism with approximately the same error.
\begin{lemma}[\cite{BhaskaraDKT12}] \label{thm:indep-reduction}
Let $\qmat \in \R^{\queries \times \uni} $ be a workload matrix.  For every $(\priv,\privd)$-differentially private mechanism $\mech$, there exists a $(2\priv,2e^{\eps}\delta)$-differentially private data-independent mechanism $\mech'$ such that
$$
\err^{\ell_2^2}(\mech',\qmat, \dsize) \leq \frac{1}{\dsize} \max_{m \in \mathbb{N}} (m \cdot \err^{\ell_2^2}(\mech, \qmat, m))
$$
\end{lemma}

As an immediate, corollary, lower bounds for data-independent mechanisms imply lower bounds for arbitrary data-dependent mechanisms for some dataset size $n^*$.  Thus we obtain the following theorem by combining Theorem~\ref{thm:indep-lb} with Lemma~\ref{thm:indep-reduction}.

\begin{theorem}\label{thm:central-err}
Let $\queries$ be linear queries with workload matrix $W \in \R^{\queries \times \uni}$.  Then for every $\priv,\privd$ smaller than some absolute constants, there exists $n^* \in \mathbb{N}$ such that
$$
\forall \dsize \leq \dsize^*~~\err_{\priv,\privd}^{\ell_2^2}(\queries,\dsize) \geq \frac{\fnorm(\qmat)}{C \priv \dsize}.
$$
\end{theorem}

By standard transformations (see e.g.~\cite{BunUV14}), we can convert this to the following sample complexity lower bound,
\begin{corollary}\label{cor:central-sc}
Let $\queries$ be linear queries with workload matrix $W \in \R^{\queries \times \uni}$.  Then for every $\priv,\privd$ smaller than some absolute constants, there exists $\alpha^* > 0$ such that
$$
\forall \alpha \leq \alpha^*~~\sc_{\priv, \privd}^{\ell_2^2}(\queries, \alpha) \geq \frac{\fnorm(\qmat)}{C \priv \alpha}
$$
\end{corollary}


%

We remark that Theorem~\ref{thm:indep-lb}, Theorem~\ref{thm:central-err}, and Corollary~\ref{cor:central-sc} can be extended to $\ell_\infty^2$-error, defined by 
\[
    \err^{\ell^2_\infty}(\mech,\queries, n) = \max_{\ds \in \uni^n}
    \ex{\mech}{\| \mech(\ds) - \queries(\ds) \|^2_\infty}^{1/2},
\]
with $\err_{\priv,\privd}^{\ell_2^2}(\queries,\dsize)$, and $\sc_{\priv, \privd}^{\ell_2^2}(\queries, \alpha)$ defined analogously, and $\fnorm(\qmat)$ replaced by $\fnorminf(\qmat)$ in the lower bounds. 

\renewcommand\thesubsection{\Alph{subsection}}

\section*{Appendix}
\addcontentsline{toc}{section}{\protect\numberline{}Appendix}

\setcounter{subsection}{0}
\subsection{KL-divergence bound}\label{ap:kl-div}

\begin{proof}[Proof of Lemma~\ref{lm:kl-div}]
  We have
    \begin{align*}
        \div(\trans(\dista_\pi^\dsize)\|\trans(\distb_\pi^\dsize)) 
        &\le\ex{V \sim \pi}{
            \div(\trans(\dista_V^\dsize)\|\trans(\distb_V^\dsize))},
        && \text{by convexity} \\
        &= \ex{V \sim \pi}{ \sum_{i =1}^\dsize 
            \div(\mech_i(\dista_V)\|\mech_i(\distb_V))},
        && \text{by independence} \\
        &= \sum_{i=1}^\dsize \ex{V \sim \pi}{
        \div(\mech_i(\dista_V)\|\mech_i(\distb_V))}\\
        &\le \sum_{i=1}^{\dsize} \ex{V \sim \pi}{
        \chd(\mech_i(\dista_V)\|\mech_i(\distb_V))},
    \end{align*}
    where the last inequality follows from the fact
    that $\chi^2$-divergence is always an upper bound
    on KL-divergence~\cite{GS-choosing}.
    Hence it suffices to show
    \begin{equation} \label{eq:expchi}
        \ex{V \sim \pi}{\chd(\mech_i(\dista_V)\|\mech_i(\distb_V))} =
        O(\priv^2)\cdot
        \max_{f \in \R^\uni: \|f\|_\infty \le 1} \ex{V \sim \pi}{\left(\ex{x \sim \dista_V}{f_x} - \ex{x \sim \distb_V}{f_x}\right)^2}.
    \end{equation}
    To that end, fix some $i$ and let $r(z | x)$ denote
    $\Pr_{\mech_i}(\mech_i(x)=z)$.  Also, let $a_v(z) = \ex{x \sim
      \dista_v}{r(z | x)}$ and let $b_v(z) = \ex{x \sim \distb_v}{ r(z |
    x)}$. Let us assume, without loss of generality, that the range
    $\Omega$ of $\mech_i$ is finite. Then, by applying the definition of
    $\chi^2$-divergence, the right hand side of \eqref{eq:expchi} may
    be rewritten as
    \begin{align}
        \ex{V \sim \pi}{ \ex{Z \sim \mech_i(\distb_V)}{ \left( \frac{b_V(Z) - a_V(Z) }{b_V(Z)} \right)^2 }}
        = \ex{V \sim \pi}{ \sum_{z \in \Omega} \left( \frac{b_V(z) - a_V(z) }{b_V(z)} \right)^2 \cdot b_V(z) }\label{eq:expchi-sum}
    \end{align}
    Let $\dist_0$ be the uniform distribution on $\uni$ 
    (any other distribution will also work).
    Then let $u(z) = \ex{x \sim \dist_0}{ r(z | x)}$.
    Since privacy implies $\frac{u(z)}{b_v(z)} \le e^\priv$,
    we may obtain an upper bound on the right hand side of
    \eqref{eq:expchi-sum} as follows.
    \begin{align}
        \ex{V \sim \pi}{ \sum_{z \in \Omega} \left( \frac{b_V(z) - a_V(z) }{b_V(z)} \right)^2 \cdot b_V(z) }
        &= \ex{V \sim \pi}{ \sum_{z \in \Omega} \left( \frac{b_V(z) - a_V(z) }{u(z)} \right)^2 \cdot \frac{u(z)^2}{b_V(z)}} \notag\\
        &\le e^\priv \cdot \ex{V \sim \pi}{ \sum_{z \in \Omega} \left( \frac{b_V(z) - a_V(z) }{u(z)} \right)^2 \cdot u(z) }\notag\\
        &= e^\priv \cdot \ex{V \sim \pi}{ \ex{z \sim \mech_i(\dist_0)}{ \left( \frac{b_V(z) - a_V(z) }{u(z)} \right)^2}} \notag\\
        &= e^\priv \cdot \ex{z \sim \mech_i(\dist_0)}{ \ex{V \sim \pi}{ \left(
          \frac{b_V(z) - a_V(z) }{u(z)} \right)^2}} \label{eq:expchi-uni}
    \end{align}
    By taking $f^z \in \R^\uni$
    to be given by $f^z_x = \frac{r(z|x)}{u(z)} - 1$,
    then we obtain
    \[
        \frac{b_V(z) - a_V(z) }{u(z)}
        = \ex{x \sim \distb_V}{ f^z_x} - \ex{x \sim \dista_V}{ f^z_x}.
    \]
    Furthermore,
    $\frac{r(z|x)}{u(z)} \le e^\priv$ is implied by privacy,
    from which it follows that $\|f^z\|_\infty \le e^\priv - 1$.
    This gives us the following bound on \eqref{eq:expchi-uni}.
    \begin{align*}
        &e^\priv \cdot \ex{z \sim \mech_i(\dist_0)}{ \ex{V \sim \pi} {\left( \frac{b_V(z) - a_V(z) }{u(z)} \right)^2 }} \\
        ={} &e^\priv \cdot \ex{z \sim \mech_i(\dist_0)}{ \ex{V \sim \pi}{ \left(
            \ex{x \sim \distb_V}{ f^z_x} - \ex{x \sim \dista_V}{ f^z_x}
        \right)^2}} \\
        \le{} &e^\priv \cdot \ex{z \sim \mech_i(\dist_0)}{ \sup_{\|f\|_\infty \le e^\priv - 1} \ex{V \sim \pi}{ \left(
          \ex{x \sim \distb_V}{ f_x }- \ex{x \sim \dista_V}{ f_x}
        \right)^2}} \\
        \le{} &e^\priv (e^\priv - 1)^2 \cdot \sup_{\|f\|_\infty \le 1} \ex{V \sim \pi}{ \left(
            \ex{x \sim \distb_V}{ f_x} - \ex{x \sim \dista_V}{ f_x}
        \right)^2}.
    \end{align*}
    Using the fact that $e^\priv (e^\priv - 1)^2 = O(\priv^2)$,
    then putting everything together, we get
    \[
        \div(\trans(\dista_\pi^\dsize)\|\trans(\distb_\pi^\dsize)) \le O(n \priv^2)\cdot 
        \max_{f \in \R^\uni: \|f\|_\infty \le 1} \ex{V \sim \pi}{
            \left(\ex{x \sim \dista_V}{f_x} - \ex{x \sim \distb_V}{f_x}\right)^2
        }.
    \]
    To obtain our result in terms of the matrix
    $M \in \R^{[K] \times \uni}$ given by $m_{v, x} = (\dista_v(x) - \distb_v(x))$,
    we note that the entries of $Mf$, indexed by $v \in [K]$,
    are given by 
    \[
        (Mf)_v
        = \sum_{x \in \uni} f_x m_{v,x}
        = \sum_{x \in \uni} (\dista_v(x)  f_x - \distb_v(x) f_x)
        = \ex{x \sim \dista_v}{f_x} - \ex{x \sim \distb_v}{f_x}.
    \]
    Hence
    \[
        \|Mf\|_{L_2(\pi)}^2
        = \sum_{v \in [K]}
            \pi(v) (Mf)_v^2
        = \ex{V \sim \pi}{
            \left( \ex{x \sim \dista_V}{f_x} - \ex{x \sim \distb_V}{f_x}\right)^2}.
    \]
    Finally, we have
    \begin{align*}
        \|M\|_{\ell_\infty \to L_2(\pi)}^2
        &= \sup_{f \in \R^\uni: \|f\|_\infty \le 1} {\|Mf\|_{L_2(\pi)}^2} \\
        &= \max_{f \in \R^\uni: \|f\|_\infty \le 1} \ex{V \sim
          \pi}{\left(\ex{x \sim \dista_V}{f_x} - \ex{x \sim \distb_V}{f_x}\right)^2},
    \end{align*}
    and this completes the proof.
\end{proof}

\subsection{Duality}\label{ap:duality}

\begin{proof}[Proof of Lemma \ref{lm:duality}]
Note that $\frac{\qmat\bullet U - \alpha\|U\|_1}{\fnorminf^\ast(U)}$ is
scale-free, and 
\[
\max_{U \neq 0} \frac{\qmat\bullet U - \alpha\|U\|_1}{\fnorminf^\ast(U)}
= \max \{\qmat\bullet U - \alpha\|U\|_1: \fnorminf^\ast(U) = 1\}.
\]
Since the set $\{U:\fnorminf^\ast(U) = 1\}$ is compact, the maximum is
achieved. Let us then define $t = \max \frac{\qmat\bullet U -
  \alpha\|U\|_1}{\fnorminf^\ast(U)}$ for the rest of the proof. 

Let us first check that $t \le \fnorminf(\qmat, \alpha)$. Let $U \neq 0$, and let $\widetilde{\qmat}$ achieve
$\fnorminf(\widetilde{\qmat}) = \fnorminf(\qmat, \alpha)$ and $\|\qmat -
\widetilde{\qmat}\|_{1 \to \infty} \le \alpha$. Then
\begin{align*}
    \qmat\bullet U &= \widetilde{\qmat} \bullet U + (\qmat-\widetilde{\qmat})\bullet U\\
    &\le \fnorminf(\widetilde{\qmat}) \fnorminf^\ast(U) + \|\qmat - \widetilde{\qmat}\|_{1 \to \infty} \|U\|_1\\
    &\le \fnorminf(\qmat, \alpha) \fnorminf^\ast(U) + \alpha \|U\|_1.
\end{align*}
The first inequality follows by the trivial case of H\"older's inequality, and the definition of $\fnorminf^\ast$. Rearranging shows that $t \le \fnorminf(\qmat, \alpha)$. 

Let us now show the harder direction, $t \ge
\fnorminf(\qmat,\alpha)$. Suppose this was false, and we had $\fnorminf(\qmat,
\alpha) > t$. We will show this implies that there exists a $U
\neq 0$ such that $\frac{\qmat\bullet U -
  \alpha\|U\|_1}{\fnorminf^\ast(U)} > t$, a contradiction. Let $S =
\{B \in \R^{\qsize \times \usize}: \fnorminf(B) \le t\}$ and $T = \{C
\in \R^{\qsize \times \usize}: \|\qmat - C\|_{1 \to \infty} \le
\alpha\}$. Then $\fnorminf(\qmat, \alpha) > t$ equivalently means $S \cap
T = \emptyset$. Since both $S$ and $T$ are convex and compact, and
$S\cap T = \emptyset$, the hyperplane separator
theorem~\cite[Corollary~11.4.2]{Rockafellar} implies that there is a
hyperplane separating them, i.e.~there is a matrix $U \in \R^{\qsize\times\usize}\setminus\{0\}$ such that
\begin{equation}\label{eq:separator}
\max\{B\bullet U: B \in S\} < \min\{C \bullet U: C \in T\}.
\end{equation}
The left hand side equals $t\fnorminf^\ast(U)$, by definition. The right hand side equals 
\begin{align*}
    \min\{\qmat \bullet U - (\qmat - C)\bullet U: C \in T\}
    &= 
    \qmat\bullet U - \max\{(\qmat - C)\bullet U: C \in T\}\\
    &= \qmat\bullet U - \max\{E\bullet U: \|E\|_{1\to\infty} \le \alpha\}\\
    &= \qmat \bullet U - \alpha\|U\|_1,
\end{align*}
where the last equality again uses the trivial case of H\"older's inequality. Therefore, \eqref{eq:separator} is equivalent to $t\fnorminf^\ast(U) <  \qmat \bullet U - \|U\|_1$, which is what we wanted to prove.
\end{proof}

\begin{proof}[Proof of Lemma~\ref{lm:gammastar-inftyto2}]
    By Lemma~\ref{lm:gammastar-dual},
    there exist diagonal matrices
    $P \in \R^{\qsize \times \qsize}$
    and $Q \in \R^{\usize \times \usize}$,
    satisfying $\tr(P^2) = \tr(Q^2) = 1$, and a matrix $\widetilde{U}$ 
    such that $U = P\widetilde{U}Q$ and 
    $\|\widetilde{U} \|_{2 \to 2} \le \gamma_2^\ast(U) $.   
    Define $S = \{i: p_{ii}^2 \le \frac{2}{\qsize}\}$.
    Then Markov's inequality shows that $|S| \ge \frac{\qsize}{2}$.
    Furthermore,
    \[
    \gamma_2^\ast(U) \ge \|\widetilde{U} \|_{2 \to 2}
    \ge 
    \|\Pi_S \widetilde{U} \|_{2 \to 2}
    \ge 
    \sqrt{\frac{\qsize}{2}}\|\Pi_S P \widetilde{U} \|_{2 \to 2},
    \]
    where the first inequality follows because multiplying by a
    projection matrix can only decrease the $\|\cdot\|_{2 \to 2}$ norm
    of the matrix, and the second inequality follows by the definition of
    $S$. 

    To finish the proof, we observe that $\|\Pi_S P \widetilde{U}\|_{2
      \to 2} \ge \|\Pi_S P \widetilde{U}Q \|_{\infty \to 2} = \|\Pi_S
    U \|_{\infty \to 2}$.  Indeed, for any $x\in \R^{\usize}$, we have
    \[
    \Pi_S P \widetilde{U}Q x
    \le 
    \|\Pi_S P \widetilde{U}\|_{2 \to 2} \|Qx\|_2
    \le \|\Pi_S P \widetilde{U}\|_{2 \to 2} \sqrt{\tr(Q^2)} \|x\|_\infty
    = \|\Pi_S P \widetilde{U}\|_{2 \to 2}  \|x\|_\infty,
    \]
    where the second inequality follows from H\"older's inequality.
\end{proof}

\begin{proof}[Proof of Lemma~\ref{lm:gammastar-inftyto2-nonuniform}]

    Without loss of generality,
    we may assume $\pi$ takes rational values.
    In particular,
    let $\widetilde{\qsize} \in \Z$
    be such that $\pi(i) \cdot \widetilde{\qsize} \in \Z$
    for all $i \in [\qsize]$.
    Then $M$ and $\pi$ may be used to define the matrix
    $\widetilde{M} \in \R^{\widetilde{\qsize} \times \uni}$,
    obtained by taking, for each $i \in [\qsize]$,
    $\pi(i) \cdot \widetilde{\qsize} $ copies of row $i$ from $M$.

    By Lemma~\ref{lm:gammastar-inftyto2},
    there exists a set $S \subseteq [\widetilde{\qsize}]$, $|S| \ge \frac{\widetilde{\qsize}}{2}$,
    such that
    $
    \sqrt{\frac{\widetilde{\qsize}}{2}} \| \Pi_S \widetilde{M} \|_{\infty \to 2}
        \le \gamma_2^*(\widetilde{M})
    $.
    Use $S$ to define the function $\widetilde{\pi}:[\qsize] \to [0,1]$
    where $\widetilde{\qsize} \widetilde{\pi}(i)$ is the number of rows selected from $\widetilde{M}$ by $S$
    which correspond to row $i$ from $M$.
    Since $|S| \ge \frac{\widetilde{\qsize}}{2}$,
    then $\sum_{i \in [\qsize]} \widetilde{\pi}(i) \ge \frac{1}{2}$.

    The equality
    $
    \sqrt{\widetilde{\qsize}} \| M \|_{\ell_\infty \to L_2(\widetilde{\pi})}
        = \| \Pi_S \widetilde{M} \|_{\infty \to 2}
    $
    follows because,
    for all $f \in \R^\uni$,
    \begin{align*}
        \| \Pi_S \widetilde{M} f \|_2^2
        = \sum_{i \in S} ( \widetilde{M}_{v,*} \cdot f )^2 
        = \widetilde{\qsize} \sum_{v \in [\qsize]} \widetilde{\pi}(v) ( M_{v,*} \cdot f )^2 
        = \widetilde{\qsize} \| M f \|_{L_2(\widetilde{\pi})}^2
    \end{align*}
    where we have denoted row $v$ of $\widetilde{M}$ by $\widetilde{M}_{v,\ast}$.

    To show 
    $  
        \gamma_2^*(\widetilde{M}) \le \widetilde{\qsize} \gamma_2^*(U)
    $,
    let $\widetilde{y}_1,\dots,\widetilde{y}_\qsize,z_1,\dots,z_\usize$ be unit vectors
    in $\ell_2$-norm which satisfy
    \[
        \gamma_2^*(\widetilde{M})
        = \sum_{i=1}^{\widetilde{\qsize}}\sum_{j = 1}^\usize \widetilde{m}_{i,j} \widetilde{y}_{i}^\top z_j,
    \]
    as they are guaranteed to exist by Lemma~\ref{lm:duality}.
    Without loss of generality, we may assume
    that if two rows $i$ and $i'$ of $\widetilde{M}$ have identical entries,
    then $\widetilde{y}_i = \widetilde{y}_{i'}$.
    Now, for all $i \in [\qsize]$, let $y_{i} = \widetilde{y}_{\widetilde{i}}$,
    where $\widetilde{i}$ is one of the rows of $\widetilde{M}$
    which was copied from row $i$ in $M$.
    Then
    \begin{align*}
        \gamma_2^\ast(\widetilde{M})
        = \sum_{i=1}^{\widetilde{\qsize}}\sum_{j = 1}^\usize \widetilde{m}_{i,j} \widetilde{y}_{i}^\top z_j
        = \widetilde{\qsize} \sum_{i = 1}^\qsize \sum_{j = 1}^\usize
            \pi(i) m_{i,j}  y_i^\top z_j
        = \widetilde{\qsize} \sum_{i = 1}^\qsize \sum_{j = 1}^\usize
            u_{i,j}  y_i^\top z_j
        \le \widetilde{\qsize} \gamma_2^*(U).
    \end{align*}

    Altogether, this gives
    \[
        \| M \|_{\ell_\infty \to L_2(\widetilde{\pi})}
        = \sqrt{\frac{2}{\widetilde{\qsize}}} \| \Pi_S \widetilde{M} \|_{\infty \to 2}
        \le \frac{2}{\widetilde{\qsize}} \gamma_2^*(\widetilde{M})
        \le 2 \gamma_2^*(U).
    \]

    Finally,
    normalize $\widetilde{\pi}$ to obtain
    the probability distribution $\widehat{\pi}$,
    given by ${\widehat{\pi}(i) = \frac{\widetilde{\pi}(i)}{\sum_i \widetilde{\pi}(i)}}$.
    Since $\sum_i \widetilde{\pi}(i) \ge \frac{1}{2}$,
    then $\widehat{\pi}(i) \le 2 \widetilde{\pi}(i)$.
    This implies
    $
        \| M \|_{\ell_\infty \to L_2(\widehat{\pi})}
        \le 2 \| M \|_{\ell_\infty \to L_2(\widetilde{\pi})}
    $,
    from which we get
    $
        \| M \|_{\ell_\infty \to L_2(\widetilde{\pi})}
        \le 4 \gamma_2^*(U)
    $.
    Since the rows that were copied from $M$ to obtain $\widetilde{M}$
    corresponded to those rows of $U$
    which were assigned non-zero probability by $\pi$,
    then, by the definitions of $\widetilde{pi}$ and $\widehat{\pi}$,
    it follows that the support of $\widehat{\pi}$
    is a subset of the support of $\pi$.
\end{proof}

\subsection{Symmetrization}\label{ap:symmetrization}

\begin{proof}[Proof of Lemma~\ref{lm:symtonotsym}]

    Let $\queries^-$ be as given by Definition~\ref{def:symmetric}.
    We consider the mechanism $\mech^-$ which,
    given a dataset $\ds^- \in (\uni^-)^n$,
    treats each element $-x$ of $\ds^-$
    as if it was the corresponding element $x$ of $\uni^+$.
    By taking the additive inverse of the answer
    $\mech^+$ produces to each query $q^+ \in \queries^+$,
    then $\mech^-$ obtains an approximation 
    for each corresponding query $q^- \in \queries^-$.
    In particular, $\errinf(\mech^+,\queries^+,n) \le \alpha$
    implies $\errinf(\mech^-,\queries^-,n) \le \alpha$.
    
    We wish to define a mechanism $\mech'$
    which takes an arbitrary dataset $\ds \in \uni^{n'}$, $n' \ge n$,
    as input and estimates the queries $\queries$.
    Let $\ds^+$ and $\ds^-$ partition $\ds$
    into elements of $\uni^+$ and $\uni^-$ respectively, and suppose
    that $\ds^+$ has size $\dsize_+$, and $\ds^-$ has size $\dsize_-$.
    Note
    \[
        \queries(\ds)
        = \frac{\dsize_+}{\dsize}  \cdot \queries^+(\ds^+) + \frac{\dsize_-}{\dsize} \cdot \queries^-(\ds^-)
    \]
    Consider the mechanism $\widetilde{\mech}^+$
    which, on the dataset $\ds$, follows $\mech^+$
    with the following modification:
    any agent holding a sample $x \in \uni^-$
    will treat their sample as if it were $x_0$,
    where $x_0$ is an arbitrary fixed element of $\uni^+$;
    any agent holding a sample $x \in \uni^+$
    treats their sample as usual.
    In other words,
    if $Y^+$ denotes the dataset obtained from $\ds$
    by replacing each element of $\uni^-$ with $x_0$,
    then $\widetilde{\mech}^+(\ds)$ has the same distribution as $\mech^-(Y)$.
    Hence,
    \[
        \E_{\widetilde{\mech}^+}[\|\queries^+(Y^+) - \widetilde{\mech}^+(\ds)\|_\infty]
        = \E_{\mech}[\|\queries^+(Y^+) - \mech(Y^+)\|_\infty]
        \le \alpha.
    \]

    Similarly, we may consider the mechanism $\widetilde{\mech}^-$
    which takes $\ds$ as input and follows $\mech^-$
    with the modification that any agent
    who holds a sample from $\uni^-$
    treats it as if it were $-x_0$.
    Taking $Y^-$ to denote the dataset obtained from $\ds$
    by replacing each element of $\uni^+$ with $-x_0$,
    then $\widetilde{\mech}(\ds)$ has the same distribution as $\mech^-(Y)$
    and hence
    \[
        \E_{\widetilde{\mech}^-}[\|\queries^-(Y^-) - \widetilde{\mech}^-(\ds)\|_\infty]
        = \E_{\mech}[\|\queries^-(Y^-) - \mech(Y)\|_\infty]
        \le \alpha.
    \]

    Furthermore, we have
    \[
        \queries(Y^+) - \queries(Y^-)
        = \frac{\dsize_+}{\dsize} \cdot \queries^+(\ds^+) + \frac{\dsize_-}{\dsize} \cdot \queries^-(\ds^-)
            + \left( \frac{ 2\dsize_+ }{ \dsize } - 1 \right) \cdot \queries(x_0)
    \]
    which motivates us to obtain an approximation of
    ${ \dsize_+ }/{ \dsize }$.
    Since this is just a single counting query, we can estimate with a
    standard use of the randomized response mechanism (see e.g. ~\cite{KasiviswanathanLNRS08}). Let us denote
    the corresponding mechanism $\mech^0$, and let $n' =
    O\left(\frac{1}{\priv^2\alpha^2}\right)$ be the number of data
    points needed so that $\E[|\mech^0(\ds)  - \frac{\dsize_+}{\dsize}|]
    \le \alpha$.

    Finally, we may define our mechanism $\mech'$
    which runs $\widetilde{\mech}^+$, $\widetilde{\mech}^-$ and $\widetilde{\mech}^0$
    in parallel and returns
    \[
        \mech'(\ds)
        = \widetilde{\mech}^+(\ds) + \widetilde{\mech}^-(\ds) - (2 \widetilde{\mech}^0(\ds) - 1) \cdot \queries(x_0).
    \]
    Since each of the mechanisms
    $\widetilde{\mech}^+$, $\widetilde{\mech}^-$ and $\widetilde{\mech}^0$
    satisfies $\eps$-differential privacy,
    then $\mech'$ satisfies $3\eps$-differential privacy by composition.
    Furthermore, by the triangle inequality,
    \begin{multline}
      \E_\mech'[\|\queries(\ds) - \mech'(\ds)\|_\infty] 
      \le \E_{\mech^+} [\|\queries(Y^+) - \mech^+(\ds)\|_\infty]
            + \E_{\mech^-} [\|\queries(Y^-) - \mech^-(\ds)\|_\infty] \\
            + 2 \cdot \E_{\mech^0} \left[ \left| \frac{\dsize_+}{\dsize} - \mech^0(\ds) \right|\right],
    \end{multline}
    where the first two terms are bounded by $\alpha$, and the final
    term by $2\alpha$. This gives error of at most $4\alpha$ in total,
    and completes the proof.
\end{proof}

\begin{proof}[Proof of Lemma~\ref{lm:symdual}]
  Since the $\|\cdot\|_{2\to\infty}$ and $\|\cdot\|_{1 \to 2}$ norms
  are both non-increasing with respect to taking submatrices, the same
  holds also for the $\fnorminf$ norm, and, therefore,
  $\fnorminf(\qmat^+) \le \fnorminf(\qmat)$. In the reverse direction,
  if $R^+A^+ = \qmat^+$ is a factorization achieving
  $\fnorminf(\qmat)$, then $R^+ A = \qmat$, where $A$ is defined by
  $a_{\query,\dsrow} = -a_{\query,-\dsrow} = a^+_{\query, \dsrow}$ for
  any $\query \in \queries$, and any $\dsrow \in \uni^+$. Clearly, $\|A\|_{1 \to 2} = \|A^+\|_{1\to
    2}$, and, therefore, the factorization $R^+ A$ certifies
  $\fnorminf(\qmat) \le \|R^+\|_{2 \to \infty} \|A\|_{1 \to 2} = \fnorminf(\qmat^+)$. The two inequalities
  imply $\fnorminf(\qmat) = \fnorminf(\qmat^+)$. 

  Next we show that $\fnorminf(\qmat^+, \alpha) \le \fnorminf(\qmat,
  \alpha)$.  Note that if $\widetilde{\qmat}$ is the approximation of
  $\qmat$ that achieves $\fnorminf(\qmat,\alpha)$, and
  $\widetilde{\qmat}^+$ is the submatrix of $\widetilde{\qmat}$
  consisting of the columns indexed by $\uni^+$, then
  $\|\widetilde{\qmat}^+ - \qmat^+\|_{1 \to \infty} \le \alpha$, and,
  by the argument above, \[\fnorminf(\qmat,\alpha) \le \fnorminf(\widetilde{\qmat}^+) \le
  \fnorminf(\widetilde{\qmat}) =\fnorminf(\qmat,\alpha).\] 
  To show he
  reverse inequality $\fnorminf(\qmat, \alpha) \le \fnorminf(\qmat^+,
  \alpha)$, take an approximation $\widetilde{\qmat}^+$ achieving
  $\fnorminf(\qmat^+, \alpha)$, and extend it to $\widetilde{\qmat}
  \in \R^{\queries \times \uni}$ by defining
  $\widetilde{w}_{\query,\elem} = -\widetilde{w}_{\query,-\elem} = \widetilde{w}^+_{\query,\elem}$ for
  all $\query \in \queries$ and $\elem \in \uni$. Then, clearly,
  $\|\widetilde{\qmat} - \qmat\|_{1 \to \infty} \le \alpha$, and, by
  the argument above, $\fnorminf(\widetilde{\qmat}) =
  \fnorminf(\widetilde{\qmat}^+)$.
  
  Finally, the claim after ``moreover'' follows because $\qmat\bullet
  U = \qmat^+ \bullet U^+$, $\|U\|_1 = \|U^+\|_1$, and $\fnorminf^\ast(U^+)
  = \fnorminf(U)$. The only non-trivial equality is the last one. To
  see why it holds, note that, first, because $\fnorminf^\ast$ is the
  dual norm of $\fnorminf$, it is, indeed, a norm, and, by homogeneity
  and the triangle inequality, 
  \[\fnorminf^\ast(U) \le 
  \frac12 \fnorminf^\ast(U^+) + \frac12\fnorminf^\ast(U-) = 
  \frac12 \fnorminf^\ast(U^+) + \frac12\fnorminf^\ast(-U^+) = 
  \fnorminf^\ast(U^+).\] 
  In the other direction, let $V^+ \in \R^{\queries \times \uni^+}$ be
  such that $V^+ \bullet U^+ = \fnorminf^\ast(U^+)$ and $\fnorminf(V^+) =
  1$. Then, we can extend $V^+$ to $V \in \R^{\queries \times \uni}$
  by setting $v_{\query,\elem} = -v_{\query,-\elem} =
  v^+_{\query,\elem}$ for all $\query \in \queries$ and $\elem \in
  \uni^+$. By the discussion above, $\fnorminf(V) = \fnorminf(V^+)$,
  and, moreover, 
  \[
  \fnorminf^\ast(U) \ge V \bullet U = V^+ \bullet U^+ = \fnorminf^\ast(U^+).
  \]
  This completes the proof.
\end{proof}

\section*{Acknowledgments}
Part of this work was done while the authors were visiting the Simons Institute for Theory
of Computing.
We are grateful to Toniann Pitassi for many helpful discussions about
local differential privacy. AE and AN were supported by an Ontario
Early Researcher Award, and an NSERC Discovery Grant. JU was supported by NSF grants
CNS-1718088, CCF-1750640 and CNS-1816028, and a Google Faculty
Research Award. 

\bibliographystyle{alpha}
\bibliography{main,refs}
\end{document}